\documentclass[reqno]{amsart}
\renewcommand{\scshape}{\normalfont\scshape} 
\usepackage{enumitem}
\usepackage{amsmath, amsthm, amssymb}
\usepackage{comment}

\newtheorem{theorem}{Theorem}[section]
\newtheorem{lemma}[theorem]{Lemma}
\newtheorem{proposition}[theorem]{Proposition}

\theoremstyle{definition}
\newtheorem{definition}[theorem]{Definition}

\newtheorem{assumption}[theorem]{Assumption}

\theoremstyle{remark}
\newtheorem{remark}[theorem]{Remark}
\numberwithin{equation}{section}

\allowdisplaybreaks

\begin{document}

\title{Energy Diffusion in a Pinned Harmonic Chain with Free Boundary Condition and Interactions of Increasing Range}

\author{Yipeng Lu}
\address{School of Mathematics, Nanjing University, Nanjing 210093, China}
\email{yipenglu@smail.nju.edu.cn}

\author{Wei Wang}
\address{School of Mathematics, Nanjing University, Nanjing 210093, China}
\email{wangweinju@nju.edu.cn}
\thanks{This work is supported by NSFC  No.12371243.}


\date{\today}



\begin{abstract}
We consider a finite two-dimensional pinned harmonic chain with an interaction range that diverges with the system size and with conservative rotational noise. Under a joint scaling in which the interaction range diverges microscopically but vanishes macroscopically, we prove that the expected empirical energy measures converge to the unique measure-valued weak solution of the heat equation.
\end{abstract}

\maketitle

\markboth{ENERGY DIFFUSION IN A PINNED HARMONIC CHAIN}{YI PENG LU AND WEI WANG}


\section{Introduction}

A fundamental problem in nonequilibrium statistical mechanics is to derive
irreversible macroscopic transport equations from microscopic particle
dynamics. In hydrodynamic limits, the size of the microscopic system tends to infinity while space and time are suitably rescaled. Conserved quantities such as energy, momentum, and particle density are then expected to evolve to deterministic partial differential equations. In particular, when energy is the only relevant macroscopic conserved field and its transport is diffusive, the empirical energy distribution is expected to evolve according to the heat equation. General references on hydrodynamic limits include \cite{spohnLargeScaleDynamics1991,kipnisScalingLimitsInteracting1999};
for Hamiltonian systems with weak conservative noise, see
\cite{ollaHydrodynamicalLimitHamiltonian1993}.

The study of heat transport in harmonic chains goes back at least to the
work of Rieder, Lebowitz, and Lieb \cite{riederPropertiesHarmonicCrystal1967}. Their analysis of a harmonic chain coupled to heat reservoirs showed that energy transport is ballistic and that Fourier’s law does not generally hold in a purely harmonic system. This is closely related to the integrability of harmonic dynamics and the resulting absence of effective scattering mechanisms. In a broad class of one-dimensional unpinned acoustic chains, momentum conservation is associated with weakly scattered long-wavelength modes, which may lead to anomalous thermal conductivity and energy superdiffusion; see \cite{lepriThermalConductionClassical2003,dharHeatTransportLowdimensional2008}. In contrast, pinning removes momentum
conservation and the zero-frequency acoustic mode. Nevertheless, in the purely harmonic setting, pinning alone does not generate diffusive transport, because the dynamics remains linear and integrable. A standard way to produce mixing while preserving the relevant conservation laws is to perturb the Hamiltonian dynamics by conservative noise. For harmonic chains with conservative stochastic perturbations, hydrodynamic limits, results concerning Fourier's law, Green--Kubo conductivity results, and equilibrium energy fluctuations were studied, respectively, in \cite{bernardinHydrodynamicsSystemHarmonic2007}, \cite{bernardinFouriersLawMicroscopic2005},\cite{basileMomentumConservingModel2006,basileThermalConductivityMomentum2009}, and \cite{basileEnergyDiffusionHarmonic2014}. For nonlinear chains, the corresponding problem of energy diffusion in nonequilibrium systems remains substantially more difficult; see \cite{bernardinTransportPropertiesChain2011,ollaMacroscopicEnergyDiffusion2013}.

For one-dimensional unpinned harmonic chains, the macroscopic behavior is
typically superdiffusive rather than diffusive. Basile, Olla, and Spohn
derived a phonon Boltzmann equation from stochastically perturbed lattice
dynamics by means of the Wigner distribution \cite{basileEnergyTransportStochastically2010}.
Jara, Komorowski, and Olla subsequently proved a direct fractional
superdiffusive limit for an unpinned harmonic chain with conservative noise,
while recovering the usual heat equation in the pinned case
\cite{jaraSuperdiffusionEnergyChain2015}. Related long-time asymptotics for
Wigner distributions and stochastic lattice wave equations were studied in
\cite{komorowskiLongTimeLarge2012,komorowskiAsymptoticsSolutionsStochastic2013}.

The Wigner transform, originally introduced in quantum statistical mechanics~\cite{wignerQuantumCorrectionThermodynamic1932}, is particularly useful for linear lattice systems because it describes how the energy is distributed over macroscopic positions and microscopic wave numbers. Although the wave function itself contains rapidly oscillating phases, its two-point correlations admit a meaningful phase-space description. Moreover, for linear harmonic dynamics, the evolution of the Wigner distribution and of the associated auxiliary quadratic quantities closes at the level of second moments.

Many existing results concern infinite lattices, periodic systems, or interaction kernels that are fixed independently of the system size. Finite systems with nonperiodic boundary conditions require additional arguments because translation invariance is lost. Hydrodynamic limits with boundary conditions and conservative noise were studied in \cite{braxmeier-evenHydrodynamicLimitHamiltonian2014}; see also \cite{komorowskiHeatFlowPeriodically2023b} for related models with boundary driving and thermostats. Weakly coupled oscillator systems and related Green--Kubo problems were considered in \cite{liveraniFourierLawWeakly2011,bernardinGreenKuboFormulaWeakly2015}.
Finally, Canestrari, Liverani, and Olla showed that a rapidly evolving
deterministic chaotic force acting as a magnetic field may produce an
effective stochastic rotational noise and yield the heat equation for the
expected energy density \cite{canestrariHeatEquationDeterministic2026}.

\par In the present paper, we study a finite two-dimensional pinned harmonic chain with an interaction range that diverges with the system size and with conservative rotational noise. The lattice sites are indexed by $\{-N,\ldots,N\}$, and the position and momentum at each site take values in $\mathbb R^2$. The interaction is generated by a smooth, compactly supported, even, and nonpositive profile $\alpha$, and its microscopic range is of order
$\delta_N(2N+1)$.

\par The stochastic perturbation is the Brownian rotation in the
two-dimensional momentum space that appears in the effective stochastic
dynamics discussed in
\cite{canestrariHeatEquationDeterministic2026}. Owing to the
antisymmetry of the rotation matrix $J$ and the corresponding It\^o
correction, the noise preserves the kinetic energy. Together with the
harmonic dynamics, it yields exact conservation of the total energy. Since
the pinning frequency satisfies $\omega_0>0$, the system has no
zero-frequency acoustic mode of the type present in unpinned chains, and
normal energy diffusion is therefore expected on the macroscopic scale.

\par We impose a pinned chain with free boundary condition as referenced in \cite{chaudhuriHeatTransportPhonon2010}. In contrast to the harmonic chain with a fixed interaction range considered in \cite{jaraSuperdiffusionEnergyChain2015}, the interaction kernel in the present model depends on $N$. Its microscopic range
$\delta_N(2N+1)$ tends to infinity, whereas its macroscopic length
$\varepsilon_N\delta_N(2N+1)$ tends to zero. We therefore consider a joint
diffusive scaling in which, as $N\to\infty$, the rescaled spatial domain
expands to $\mathbb R$, the interaction remains local on the macroscopic
scale, and a nontrivial limit emerges on the time scale
$a_N=\bigl(\delta_N(2N+1)\varepsilon_N^2\bigr)^{-1}.$
Compared with the usual diffusive time scale $\varepsilon_N^{-2}$ for
fixed-range models, the time scale $a_N$ contains the additional factor
$\left(\delta_N(2N+1)\right)^{-1}$. This modification is caused by the growth of the interaction range, since $\int_{\mathbb T}|\omega_N'(k)|^2\,dk$ is of order $\delta_NN$. Hence, the time scale $\bigl(\delta_N(2N+1)\varepsilon_N^2\bigr)^{-1}$ is the natural one on which a finite and nontrivial macroscopic diffusion limit can be observed.

\par We allow the initial macroscopic energy profile to be a nonnegative Radon
measure $\mu_0$ on $\mathbb R$. We do not require $\mu_0$ to possess a
density or to have finite total mass; only a suitable exponentially
weighted integrability condition is imposed. Our main result states that,
under the scaling and initial assumptions specified in Section~2, the
expected empirical energy measures converge, uniformly on every compact
time interval in the vague topology, to the unique measure-valued weak
solution of the heat equation. The effective diffusivity is
$\frac{\widehat c}{4\pi^2\gamma}$, where $\widehat c$ is defined in \eqref{hatc const}.
The coefficient $\widehat c$ has a natural interpretation in terms of the
dispersion relation. Up to the normalization factor determined by our
Fourier-transform convention, it represents the limiting normalized
contribution of the squared group velocities of the microscopic
wave-number modes. It is determined by the interaction profile $\alpha$
and the pinning frequency $\omega_0$, while the factor $\gamma^{-1}$
reflects the mixing time of the rotational noise. Thus, the limiting
diffusivity incorporates the effects of the microscopic interaction, the
pinning strength, and the intensity of the stochastic rotation.
\par The proof is based on the Wigner transform. We first prove that the local
energy and the corresponding wave-function energy are asymptotically
equivalent after testing on the macroscopic scale. We then derive a closed
system of evolution equations for the Wigner transform and its auxiliary
quadratic quantities. By identifying the limiting quadratic Fourier symbol
at small macroscopic wave numbers, we obtain the weak formulation of the
heat equation. Finally, local mass estimates and time equicontinuity are
used to establish compactness and identify the limit in the space of
nonnegative Radon measures endowed with the vague topology.
\par A significant additional difficulty comes from the free boundary condition. Unlike an infinite-volume or periodic boundary condition, zero
extension destroys translation invariance in the finite system. As a
consequence, the quantity $\eta_x^N$ defined in \eqref{def:etaN} does not
vanish near the boundary and produces the boundary remainders
$\widehat{\mathcal R}_N^{(i)}$ in the Wigner equations. These remainders
cannot be shown to vanish by estimating their untested norms directly.
Instead, in Lemma~\ref{lem:yuxiang}, we first pair them with a test function
and then use integration by parts together with the rapid decay of Schwartz
functions and their Fourier transforms. This allows us to prove that the
boundary contributions vanish in the macroscopic weak limit.
\par Another difficulty is that, under the assumptions of the present paper, we
only have an estimate of the form
$\sup_{t,p}
\int_{\mathbb T}
\left|\widehat W_N(t,p,k)\right|\,dk
\leq C N\varepsilon_N,$
and the right-hand side is not uniformly bounded in $N$. This differs from
the finite-macroscopic-energy framework of
\cite{jaraSuperdiffusionEnergyChain2015}, where a uniform bound on the
Wigner distributions is available, and from the periodic system with
normalized total energy considered in~\cite{canestrariHeatEquationDeterministic2026}. Consequently, compactness cannot be deduced directly from a uniform bound on the total mass. Instead, we establish uniform local mass estimates for the empirical energy measures on $\mathbb R$ and study their convergence in the vague topology.

\par Finally, the assumptions on both the initial data and the scaling sequences
are relatively general. The initial profile $\mu_0$ is only required to be
a locally finite nonnegative Radon measure that acts finitely on a prescribed
exponentially decaying weight, and it may therefore have infinite total
mass. Moreover, the scaling assumptions are formulated in terms of
asymptotic relations and do not require $\delta_N$ and $\varepsilon_N$ to
be powers of $N$. The result consequently applies to a broader class of
mesoscopic scaling regimes than those described by polynomial sequences
alone.

\subsection{Notations}
Throughout this paper, we adopt the following standard notations and conventions.
\begin{enumerate}[label=(\roman*)]
    \item Unless otherwise specified, all limits (denoted by $\to$) and asymptotic notations (such as $O,o, \ll$) are tacitly understood in the limit as $N \to \infty$.
    \item We denote the set of nonnegative integers by $\mathbb{N} = \{0, 1, 2, \dots\}$, and the set of positive integers by $\mathbb{N_+} = \{1, 2, \dots\}$.
    \item We denote by $\mathcal{M}_+(\mathbb{R})$ the set of all nonnegative Radon measures on $\mathbb{R}$.
    
    \item We denote by $f^*$ the complex conjugate of $f$.
    \item For vectors $q_x, p_x \in \mathbb{R}^2$, we use the simplified notation $q_x p_x$ to denote their standard Euclidean inner product $q_x \cdot p_x$. Accordingly, $p_x^2$ denotes the squared Euclidean norm $|p_x|^2$.
    \item For discrete sequences $a, b \in \ell^2(\mathbb{Z}; \mathbb{R}^2)$, we denote their inner product by $\langle a, b \rangle_{\ell^2(\mathbb{Z})} := \sum_{x \in \mathbb{Z}} a_x b_x^*$. Similarly, for functions $f, g\in L^2(\Omega)$ (where $\Omega$ is $\mathbb{R}$ or $\mathbb{T}$), the $L^2$-inner product is defined as $\langle f, g \rangle_{L^2(\Omega)} := \int_{\Omega} f(x) g^*(x) \, dx$.
    \item We use $A \lesssim B$ to mean that $A \le C B$ for a generic constant $C > 0$ independent of $N$ and the test function $\varphi$. Furthermore, $A \asymp B$ indicates that $A \lesssim B$ and $B \lesssim A$.\label{note:lesssim}
    \item  $O(A)$ denotes any quantity B satisfying $|B|\lesssim A$. \label{note:O}

\end{enumerate}

\section{The model and main result}
We first introduce the interaction profile.
\begin{assumption}
Let $\alpha \in C_c^\infty(\mathbb{R})$ satisfy the following assumptions:
\begin{enumerate}[label=(a\arabic*)]
    \item \label{a1} The support of the function is compact, $\operatorname{supp}(\alpha) \subset [-M, M]$ for some $M > 0$,
\item \label{a2} $\alpha$ is an even function, which means $\alpha(u) = \alpha(-u)$ for all $u \in \mathbb{R}$,
\item \label{a3} $\alpha$ is non-positive, namely $\alpha(u) \le 0$ for all $u \in \mathbb{R}$.
\end{enumerate}
\end{assumption}
Based on the profile $\alpha$, the interaction coefficients $\alpha_x^N$ are defined as
\begin{align} \label{def:alphaN}
    \alpha_x^N :=
    \begin{cases} 
    \frac{1}{\delta_N (2N+1)} \alpha\left( \frac{x}{\delta_N (2N+1)} \right),  & x \neq 0, \\
    \omega_0^2 - \sum_{y \neq 0} \alpha_y^N, & x = 0.
    \end{cases}
\end{align}
A direct consequence of \eqref{def:alphaN} is that the sum of all coefficients yields the squared pinning frequency. That is, the zero-th Fourier coefficient satisfies
\begin{align}
    \widehat{\alpha^N}(0) = \sum_{x \in \mathbb{Z}} \alpha_x^N = \omega_0^2 > 0.
\end{align}

We consider a finite system of interacting oscillators. The dynamics of the position $q_x(t)$ and the momentum $p_x(t)$ is governed by the following system of stochastic differential equations
\begin{align} \label{eq:main_SDE}
    \begin{cases}
    dq_x(t) = p_x(t) dt, \\
    dp_x(t) = \left( Gq_x(t) - \omega_0^2 q_x(t) \right) dt - \gamma p_x(t) dt + \sqrt{2\gamma} J p_x(t) dw_x(t).
    \end{cases}
\end{align}
Here, $x \in \{-N,...N\},$ for notational convenience in discrete convolutions and Fourier transforms, we extend the position and momentum variables to the whole lattice by setting
$q_x(t)=p_x(t)=0,$ for $|x|>N.$ For each site $x$, the position and momentum are two-dimensional vectors, i.e., $q_x(t) = (q_{x,1}(t), q_{x,2}(t)) \in \mathbb{R}^2$ and $p_x(t) = (p_{x,1}(t), p_{x,2}(t)) \in \mathbb{R}^2$. For simplicity, we denote the squared Euclidean norm by $p_x^2 := |p_x|^2 = p_{x,1}^2 + p_{x,2}^2$.

The parameters $\omega_0 > 0$ and $\gamma > 0$ represent the pinning frequency and the friction coefficient, respectively. The noise is driven by a family of independent standard one-dimensional Brownian motions $\{w_x(t)\}_{x\in\mathbb{Z}}$. The matrix $J$ generates a rotation in $\mathbb{R}^2$, which is given by
\begin{align*}
    J = \begin{pmatrix} 0 & 1 \\ -1 & 0 \end{pmatrix}.
\end{align*}
The operator $G$ represents the interaction between oscillators, defined as
\begin{align*}
    Gq_x(t) = -\sum_{y=-N}^{N} \alpha_{x-y}^N (q_y(t) - q_x(t)),\qquad |x|\leq N.
\end{align*}
We define the single particle energy
\begin{equation}
    e_x= \frac{1}{2}p_x^2  - \frac{1}{4}  \sum_{y=-N}^{N}  \alpha_{x-y}^N (q_y - q_x)^2+\frac{\omega_0^2}{2} q_x^2,\qquad |x|\leq N.
\end{equation} 
For notational convenience, we extend the energy to the whole lattice by setting
$e_x=0$ for $|x|>N$\,,  and the total energy 
\begin{align*} \mathcal{E}_N &= H_N(p, q) = \sum_{x=-N}^{N} e_x \\ &= \sum_{x=-N}^{N} \frac{p_x^2}{2} - \frac{1}{4} \sum_{x=-N}^{N} \sum_{y=-N}^{N} \alpha_{x-y}^N (q_y - q_x)^2 + \frac{\omega_0^2}{2} \sum_{x=-N}^{N} q_x^2 \\ &= \frac{1}{2} \sum_{x=-N}^{N} p_x^2 + \frac{1}{2} \sum_{x=-N}^{N} \sum_{y=-N}^{N} \alpha_{x-y}^N q_x q_y  - \frac{1}{2} \sum_{x=-N}^{N} \left( \sum_{y=-N}^{N} \alpha_{x-y}^N - \omega_0^2 \right) q_x^2. \end{align*}
Set 
\begin{align} \label{def:etaN}
    \eta_x^N:= \sum_{y=-N}^{N}   \alpha_{x-y}^N-\omega_0^2,
\end{align}
then
\begin{equation}
    \mathcal E_N=\frac{1}{2}\sum_{x=-N}^{N}p_x^2+\frac{1}{2} \sum_{x=-N}^{N} \sum_{y=-N}^{N} 
    \alpha_{x-y}^Nq_xq_y-\frac{1}{2}\sum_{x=-N}^{N}\eta_x^Nq_x^2.
\end{equation}

Let
\begin{align} 
    &\omega_N(k):=\sqrt{\widehat{\alpha^N }(k)}>0,  \label{def:omega}\\
    &a_N:=\left(\delta_N(2N+1)\varepsilon_N^2 \right)^{-1}\,,\label{def:aN}
\end{align}
and for any test function $ \varphi \in C(\mathbb{R}) $,  
\begin{equation}
    I_{N}(t,\varphi) := \varepsilon_N \sum_{x=-N}^{N} \varphi(\varepsilon_N x)\mathbb{E} \left[  e_x(a_Nt) \right].
\end{equation}
\begin{assumption}\label{assum:main_conditions}
We assume the following conditions hold throughout the paper.
 \begin{enumerate}[label=(H\arabic*)]
   \item \label{H1} There exists a constant $\mathcal{E}_* > 0$ such that
$$\mathbb{E}[\mathcal{E}_N(0)] \le \mathcal{E}_* N.$$

    \item \label{H2} 
    There exist $b_0>0$ and an integer $m_0\geq 3$ such that, as $N\to\infty$, the scaling parameters satisfy
\begin{align*}
     &\delta_N N \to \infty, \ \ N \varepsilon_N^{1+1/m_0} \to \infty, \ \delta_N^2 N^3\varepsilon_N^2 \to 0,
     \\ &\delta_N^{m_0} N\to 0,\ \frac{e^{-b_0N\varepsilon_N}}{\delta_N\varepsilon_N} \to 0.
\end{align*}

\item \label{H3}
For every  test function $ \varphi \in C_c^\infty(\mathbb{R}),$ 
\begin{equation}
    \lim_{N \to \infty} I_{N}(0,\varphi) = \int_{\mathbb{R}} \varphi(u) \mu_0(du),
\end{equation}
 for some  $\mu_0 \in \mathcal{M_+}(\mathbb{R}).$
Moreover, the measure $\mu_0$ satisfies
\begin{equation} \label{H3,2}
    \int_\mathbb{R} e^{-b_0|x|}\mu_0(dx) < \infty,
\end{equation}
where $b_0$ is the same constant as in \ref{H2}.

\item \label{H4}
With $b_0$ as in \ref{H2}, there exists a constant $b_1\in(0,b_0)$
such that
\begin{align}
    \sup_N \varepsilon_N\sum_{x=-N}^{N} e^{-b_1\varepsilon_N |x|}\mathbb{E}\left[  e_x(0) \right] <\infty.
\end{align}
\end{enumerate}
\end{assumption}

\begin{remark}
The class of scaling parameters satisfying \ref{H2} is nonempty.
Indeed, regarding \ref{H2}, set $\delta_N = N^{-b}$ and $\varepsilon_N = N^{-c}$, 
with 
\[
    \frac12<b<1,\quad\frac12<c<1,\quad b+c>\frac32
\]
then \ref{H2} holds for a sufficiently large integer $m_0\geq3$.
\end{remark}

\begin{remark}
Since $\alpha$ is compactly supported, the interaction coefficients defined
in \eqref{def:alphaN} vanish whenever $|x|>M\delta_N(2N+1)$. By \ref{H2}, the microscopic interaction range is of order $\delta_N(2N+1) \to \infty$. Thus, on the microscopic lattice scale, each particle interacts with an increasing number of particles. On the other hand, $\delta_N(2N+1)/(2N+1)=\delta_N \to 0$, so the interaction range is negligible compared with the total number of particles in the system. Moreover, its macroscopic length satisfies $\varepsilon_N\delta_N(2N+1)\to 0.$ Therefore, the interaction is long-range on the microscopic scale, but remains local on the macroscopic scale.
\end{remark}

We define the class of real-valued functions
\begin{equation}
    \mathcal{C} := \left\{ \varphi \in C^\infty(\mathbb{R}) : \text{for every } m \in \mathbb{N},\ 
    \sup_{x \in \mathbb{R}}  e^{b_0|x|} |\varphi^{(m)}(x)| < \infty \right\}.
\end{equation}
Then $C_c^\infty(\mathbb{R})\subset\mathcal{C}\subset\mathcal{S}(\mathbb{R})$,
and for every $\varphi\in\mathcal{C}$ and $m\in\mathbb{N}$, there exists a constant $C_{m,\varphi}>0$ such that
\begin{equation} \label{phi_exp——decay}
    \left|\varphi^{(m)}(x)\right|
    \leq C_{m,\varphi}e^{-b_0|x|},
    \qquad x\in\mathbb{R}.
\end{equation}

\begin{remark}
It follows from Lemma \ref{lem:t=0 H2 C_function} below that the convergence in \ref{H3} extends to every $\varphi\in\mathcal{C}$.
\end{remark}

In order to formulate our main result, we introduce the measure-valued heat equation
with initial condition $\mu_0\in \mathcal{M_+}(\mathbb{R})$.
\begin{equation}\label{measure heat equation}
    \begin{cases}
\partial_t \mu(t,y) = \frac{\hat c}{4\pi^2\gamma} \partial_y^2 \mu(t,y), \quad y \in \mathbb{R},\\
\mu(0, dy) = \mu_0(dy),
\end{cases}
\end{equation} 
where  
   \begin{align} \label{hatc const}
      \hat c= \frac{1}{4} \int_{\mathbb{R}} \frac{|\hat{\alpha}'(p)|^2}{\omega_0^2 + \|\alpha\|_{L^1} + \hat{\alpha}(p)} \, dp.
   \end{align}

\begin{definition}
    We say that a map $\mu : [0, +\infty) \to \mathcal{M_+}(\mathbb{R})$ is a (weak, measure-valued) solution of \eqref{measure heat equation} if it belongs to $C([0, +\infty), \mathcal{M_+}(\mathbb{R}))$ and for every  test function $\varphi \in C_c^2(\mathbb{R})$,
\begin{equation} \label{def:slove}
    \int_{\mathbb{R}} \varphi(y) \mu(t, dy) - \int_{\mathbb{R}} \varphi(y) \mu_0(dy) = \frac{\hat c}{4\pi^2\gamma} \int_0^t ds \int_{\mathbb{R}} \varphi''(y) \mu(s, dy).
\end{equation}
\end{definition} 
\begin{theorem} \label{mainthm}
    Under the Assumptions \ref{H1}--\ref{H4}, for every  $T>0,$ test function $\varphi \in C_c(\mathbb{R})$, 
    \begin{align}
    \lim_{N \to \infty} \sup_{0 \le t \le T} \left| \varepsilon_N \sum_{x \in \mathbb{Z}} \varphi(\varepsilon_N x) \mathbb{E}[e_x(a_N t)] - \int_{\mathbb{R}} \varphi(y) \mu(t, dy) \right| = 0,
    \end{align}
    where $\mu(t,dy)$ is a weak solution of  \eqref{measure heat equation}.
\end{theorem}

\section{Asymptotics of the wave function}

\begin{lemma} \label{lem:conservation of energy}
For all $t>0$, $\mathcal{E}_N(t)=\mathcal{E}_N(0)$.
\end{lemma}
\begin{proof}
    For any sufficiently smooth function $f(q,p)$, Itô's formula yields
\begin{align}\label{eq:itof}
     df(q,p) &= \sum_{x=-N}^N \left( \frac{\partial f}{\partial q_x} \cdot dq_x + \frac{\partial f}{\partial p_x} \cdot dp_x \right) + \gamma \sum_{x=-N}^N Jp_x \cdot \frac{\partial^2 f}{\partial p_x^2} Jp_x dt
 \notag \\&=\sum_{x=-N}^N \Bigg[ p_x \cdot \frac{\partial f}{\partial q_x} dt + \frac{\partial f}{\partial p_x} \cdot \bigg( -(\alpha^N * q)_x dt + \eta_x^N q_x dt \notag \\&\quad  - \gamma p_x dt  + \sqrt{2\gamma} Jp_x dw_x(t) \bigg) + \gamma Jp_x \cdot \frac{\partial^2 f}{\partial p_x^2} Jp_x dt \Bigg]
 \notag \\&=\sum_{x=-N}^N \left( p_x \cdot \frac{\partial f}{\partial q_x} - (\alpha^N * q)_x \cdot\frac{\partial f}{\partial p_x} \right)dt + \sum_{x=-N}^N \eta_x^N q_x \cdot\frac{\partial f}{\partial p_x}dt
  \notag  \\&\quad + \gamma \sum_{x=-N}^N \left( Jp_x \cdot \frac{\partial^2 f}{\partial p_x^2} Jp_x - p_x \cdot\frac{\partial f}{\partial p_x} \right)dt
  +\sqrt{2\gamma} \sum_{x=-N}^N Jp_x\cdot\frac{\partial f}{\partial p_x} dw_x(t) .
\end{align}
Choose $f=\mathcal{E}_N$. Since $p_x\cdot Jp_x=0,\ Jp_x\cdot Jp_x=p_x^2$, we have
\begin{align}
    d\mathcal{E}_N(t)&=\sum_{x=-N}^N \bigg[ \frac{1}{2}\sum_{y=-N}^N \alpha_{x-y}^N(q_y-q_x) p_x+\omega_0^2 q_xp_x-\frac{1}{2}\sum_{y=-N}^N \alpha_{x-y}^N(q_y-q_x) p_y
    \notag \\& \quad -  \sum_{y=-N}^N \alpha_{x-y}^Nq_y p_x  \bigg]dt  
    + \sum_{x=-N}^N \eta_x^N q_x p_xdt  
\notag \\& \quad  + \gamma \sum_{x=-N}^N \left( Jp_x \cdot  Jp_x - p_x \cdot p_x \right)dt+\sqrt{2\gamma} Jp_x\cdot p_x dw_x(t) 
  \notag \\&=-\frac{1}{2}\sum_{x=-N}^N  \sum_{y=-N}^N \alpha_{x-y}^N
  \left[(q_y+q_x) p_x+ (q_y-q_x) p_y\right] dt 
  \notag \\& \quad+ \sum_{x=-N}^N ( \eta_x^N q_x p_x+\omega_0^2 q_xp_x)dt  
  \notag \\&=-\frac{1}{2}\sum_{x=-N}^N  \sum_{y=-N}^N \alpha_{x-y}^N(q_y-q_x) (p_x+p_y)  dt.
\end{align}
Since $\alpha^N$ is even, interchanging $x$ and $y$ yields $d\mathcal{E}_N(t) = -d\mathcal{E}_N(t)$. Consequently, 
$\frac{d\mathcal{E}_N(t)}{dt} = 0,$ 
which implies $\mathcal{E}_N(t) = \mathcal{E}_N(0)$ for all $t > 0$.
\end{proof}

\begin{remark} \label{remark:no time}
   The preceding lemma establishes the conservation of the total energy. Henceforth, whenever the subsequent proofs solely involve estimates of the total energy, the time variable $t$ will be omitted for brevity.
\end{remark}
To approximate the energy, we define the wave function and its Fourier transform as (see Appendix A)
\begin{align} 
   &\psi_x=(\check{\omega}_N*q)_x+ip_x,\ x\in \mathbb{Z},\\
    &\hat{\psi_k}=\hat{\psi}(k):=\omega_N(k)\hat{q}_k+i\hat{p}_k,\ k\in \mathbb{T}.
\end{align}
Consequently,
$\left| \psi_x \right|^2=|(\check{\omega}_N*q)_x|^2+\left| p_x \right|^2.$
By the Plancherel theorem, 
\begin{align}
    \sum_{x \in \mathbb{Z}}  \left[  \frac{1}{2}\left| \psi_x \right|^2 \right]
     & = \frac{1}{2} \lVert\check{\omega}_N \ast q\rVert_{\ell^2(\mathbb{Z})}^2+ \frac{1}{2} \lVert p \rVert_{\ell^2(\mathbb{Z})}^2 \notag\\
    &= \frac{1}{2} \lVert\omega_N \hat q \rVert_{L^2(\mathbb{T})}^2+ \frac{1}{2} \lVert \hat p \rVert_{L^2(\mathbb{T})}^2 \notag\\
    &= \frac{1}{2}\int_\mathbb{T}\left(|\omega_N(k) \hat q_k|^2
    +|\hat p_k|^2\right)\,dk.
\end{align}
By Lemma \ref{lem:omega_N_properties}~(iii), and \ref{H1}
\begin{align} \label{leq:wave function energy}
    \sum_{x \in \mathbb{Z}} \mathbb{E} \left[  \frac{1}{2}\left| \psi_x \right|^2 \right] &\lesssim \mathbb{E} \left[ \frac{1}{2}\int_\mathbb{T}\left(| \hat q_k|^2
    +|\hat p_k|^2\right)\,dk\right] \notag \\
    &\lesssim \mathbb{E}[\mathcal{E}_N]\leq \mathcal{E}_* N.
\end{align}
In contrast to the periodic system considered in \cite{canestrariHeatEquationDeterministic2026} and the infinite chain on $\mathbb{Z}$ studied in \cite{jaraSuperdiffusionEnergyChain2015}, we consider a finite pinned chain with free boundary conditions. The energy $e_x$ in the present finite system is set to zero for $|x|>N$, whereas the wave function $\psi_x$ is defined on the whole lattice $\mathbb{Z}$ and need not vanish outside $[-N,N]$. Consequently, the identity
\begin{equation}
\sum_{x=-N}^{N} e_x
=
\frac{1}{2}\sum_{x\in\mathbb{Z}}|\psi_x|^2
\end{equation}
does not hold in the present setting. In fact we have 
\begin{align}
    \sum_{x=-N}^{N} e_x& =\mathcal{E}_N
    =\frac{1}{2}\sum_{x=-N}^{N}p_x^2+\frac{1}{2} \sum_{x=-N}^{N} \sum_{y=-N}^{N} 
    \alpha_{x-y}^Nq_xq_y-\frac{1}{2}\sum_{x=-N}^{N}\eta_x^Nq_x^2 \notag\\
    &=\frac{1}{2} \lVert p \rVert_{\ell^2(\mathbb{Z})}^2
    +  \frac{1}{2} \langle\alpha^N*q,q \rangle_{\ell^2(\mathbb{Z})}
    -\frac{1}{2}\sum_{x=-N}^{N}\eta_x^Nq_x^2\notag\\
    &=  \frac{1}{2}\int_\mathbb{T}\left(|\omega_N(k) \hat q_k|^2
    +|\hat p_k|^2\right)\,dk
    -\frac{1}{2}\sum_{x=-N}^{N}\eta_x^Nq_x^2.
\end{align}
Nevertheless, the extra term vanishes in the limit when tested against a suitable test function. A precise justification is given in the proof of the following proposition.

\begin{proposition} 
\label{thm:wavebijin}
For every $\varphi \in  \mathcal{C}$
, there exists $C_\varphi>0$, such that
\begin{eqnarray*}
    &&\sup_{t\geq 0}\left|\varepsilon_N \sum_{x \in \mathbb{Z}} \varphi(\varepsilon_N x) \mathbb{E} \left[  e_x(a_Nt) -\frac{1}{2}\left| \psi_x(a_Nt) \right|^2\right]\right|
   \\&=& C_\varphi O(N \varepsilon_N e^{-b_0N\varepsilon_N}) + C_\varphi O(\delta_N^2 N^3 \varepsilon_N^3).
\end{eqnarray*}
Here, the constants implicit in the $O(\cdot)$ terms are independent of both $N$ and $\varphi$. See Notations \ref{note:lesssim} and \ref{note:O} for details.
\end{proposition} 

\begin{proof}
By Lemma \ref{lem:omega_N_properties}~(i), $\omega_N(k)$ is a real-valued even function. Applying the properties of the Fourier transform for discrete convolutions (see Appendix A), we obtain
\begin{align}
    (\check{\omega}_N*q)_x^*
    &=\int_\mathbb{T}\omega_N(k)\hat{q}_k^*e^{-2\pi ikx}dk
     = \int_\mathbb{T}\omega_N(k)\hat{q}(-k)e^{-2\pi ikx}dk 
     \notag\\ &=\int_\mathbb{T}\omega_N(-k)\hat{q}(k)e^{2\pi ikx}dk
  =  (\check{\omega}_N*q)_x.
\end{align} 
Hence, $(\check{\omega}_N*q)_x \in \mathbb{R}$, which implies $\left| \psi_x \right|^2 = \left|(\check{\omega}_N*q)_x\right|^2 + p_x^2$. 
We decompose the difference as
\begin{eqnarray*}
    &&\varepsilon_N \sum_{x \in \mathbb{Z}} \varphi(\varepsilon_N x) \mathbb{E} \left[  e_x(a_Nt) -\frac{1}{2}\left| \psi_x(a_Nt) \right|^2\right] \\
&=&\varepsilon_N \sum_{x \in \mathbb{Z}} \varphi(\varepsilon_N x) \mathbb{E} \Bigg[  \frac{1}{2}p_x^2 - \frac{1}{4} \sum_{y=-N}^{N} \alpha_{x-y}^N (q_y - q_x)^2 + \frac{\omega_0^2}{2} q_x^2 
- \frac{1}{2}\left|(\check{\omega}_N*q)_x\right|^2 \\
&&\qquad\qquad\qquad- \frac{1}{2} p_x^2 \Bigg] -\varepsilon_N \sum_{|x|> N} \varphi(\varepsilon_N x) \mathbb{E} \Bigg[  - \frac{1}{4}  \sum_{y=-N}^{N}  \alpha_{x-y}^N (q_y - q_x)^2\Bigg]   \\
& =& \varepsilon_N \sum_{x \in \mathbb{Z}} \varphi(\varepsilon_N x) \mathbb{E} \Bigg[ 
\frac{1}{4}\left( \omega_0^2 - \sum_{y=-N}^{N}\alpha_{x-y}^N \right)q_x^2 + \frac{1}{4}\left(\omega_0^2q_x^2 - \sum_{y=-N}^{N}\alpha_{x-y}^Nq_y^2\right) \\
&&\qquad\qquad\qquad\qquad\qquad\qquad\qquad\qquad+ \frac{1}{2}\left( \sum_{y=-N}^{N}\alpha_{x-y}^Nq_xq_y - \left|(\check{\omega}_N*q)_x\right|^2\right) \Bigg] 
\\
&&{} +\varepsilon_N \sum_{|x|> N} \varphi(\varepsilon_N x) 
\mathbb{E} \Bigg[ \frac{1}{4} \sum_{y=-N}^{N}\alpha_{x-y}^Nq_y^2 \Bigg]
\notag\\ &=:&I_1^N+I_2^N+I_3^N+I_4^N.
\end{eqnarray*}


We now estimate these four terms separately.
\par First, we consider $I_1^N$. By Lemma \ref{lem:eta}~(ii), it suffices to consider those $x$ satisfying $|x| > N - M \delta_N (2N+1)$.
Applying Lemma \ref{lem:eta}(i) and \eqref{phi_exp——decay}, we have
\begin{align*}
    |4I_1^N| &\lesssim C_\varphi\varepsilon_N \sum_{x=-N}^{N} e^{-b_0\varepsilon_N|x|} |\eta_x^N| \mathbb{E}[q_x^2] 
    \lesssim C_\varphi N \varepsilon_N e^{-b_0N\varepsilon_N} \frac{1}{N}\sum_{x=-N}^{N} \mathbb{E}[e_x]
    \\ &\lesssim C_\varphi N \varepsilon_N e^{-b_0N\varepsilon_N}.
\end{align*}
\par Next, we estimate $I_2^N$. For any $\varphi \in \mathcal{C}$ and $m \in \mathbb{N}$, we have $\left\|\frac{d^m}{dx^m}\varphi\right\|_\infty < \infty$.  Utilizing the fact that $\alpha^N$ is an even sequence and interchanging the roles of $x$ and~$y$, we obtain
\begin{align*} 
    4I_2^N \quad&=\quad \varepsilon_N \sum_{x \in \mathbb{Z}} \varphi(\varepsilon_N x) \mathbb{E}\left[ \omega_0^2 q_x^2 - \sum_{y \in \mathbb{Z}} \alpha_{x-y}^N q_y^2 \right] \\ 
    &=\quad \varepsilon_N \sum_{x \in \mathbb{Z}} \varphi(\varepsilon_N x) \mathbb{E}[\omega_0^2 q_x^2] - \varepsilon_N \sum_{x,y \in \mathbb{Z}} \varphi(\varepsilon_N y) \alpha_{x-y}^N \mathbb{E}[q_x^2] \\ 
    &=\quad \varepsilon_N \sum_{x=-N}^{N} \mathbb{E}[q_x^2] \left( \omega_0^2 \varphi(\varepsilon_N x) - \sum_{y \in \mathbb{Z}} \alpha_y^N \varphi(\varepsilon_N(x+y)) \right) \\ 
    &=\quad \varepsilon_N \sum_{x=-N}^{N} \mathbb{E}[q_x^2] \sum_{y \in \mathbb{Z}} \alpha_y^N \left( \varphi(\varepsilon_N x) - \varphi(\varepsilon_N(x+y)) \right). 
\end{align*} 
Notice that $\alpha_y^N = 0$ and  $|\varepsilon_N y| \lesssim \delta_N N \varepsilon_N \to 0$ as $N \to \infty$  for $|y| > M \delta_N(2N+1)$\,.  This justifies the Taylor expansion
$$
    \varphi(\varepsilon_N(x+y)) - \varphi(\varepsilon_N x) = \varphi'(\varepsilon_N x)\varepsilon_N y + \frac{1}{2}\varphi''(\varepsilon_N x)(\varepsilon_N y)^2 + O\left(\left\|\varphi'''\right\|_\infty \varepsilon_N^3 |y|^3\right). 
$$
Since $\alpha^N$ is an even sequence, $\sum_{y \in \mathbb{Z}} \alpha_y^N y = 0$. By Lemma~\ref{lem:alpha_properties}~(iii),
 $$\sum_{y \in \mathbb{Z}} |y|^m |\alpha_y^N| = C_m(\delta_N N)^m + O\left((\delta_N N)^{m-1}\right).$$
  Substituting this back into the expression for $4I_2^N$ and applying \ref{H1}, we get
\begin{align*} 
    |-4I_2^N| &= \varepsilon_N \left|\sum_{x=-N}^{N} \mathbb{E}[q_x^2] \left( \varphi''(\varepsilon_N x) \frac{\varepsilon_N^2}{2} \left( C(\delta_N N)^2 + O(\delta_N N) \right) \right)\right| \\ 
    &\lesssim \delta_N^2 N^2 \varepsilon_N^3 \sum_{x=-N}^{N} \mathbb{E}[q_x^2] \left|\varphi''(\varepsilon_N x)\right| \\
    &\lesssim C_\varphi\delta_N^2 N^3 \varepsilon_N^3 \cdot \frac{1}{N} \mathbb{E}[\mathcal{E}_N] 
    = C_\varphi O(\delta_N^2 N^3 \varepsilon_N^3). 
\end{align*} 
\par We next estimate $I_3^N$. Rewrite the term $I_3^N$\,,  by using discrete convolutions,  as
\begin{align*}
    I_3^N &= \varepsilon_N \sum_{x \in \mathbb{Z}} \varphi(\varepsilon_N x) \mathbb{E} \left[ \frac{1}{2} (\alpha^N * q)_x q_x - \frac{1}{2} |(\check{\omega}_N*q)_x|^2 \right] \\
    &=: I_{3,1}^N + I_{3,2}^N
\end{align*}
with  
\begin{eqnarray*}
    I_{3,1}^N &=& \frac{\varepsilon_N}{2} \int_{\mathbb{R} \times \mathbb{T}} \hat{\varphi}^*(p) \widehat{\alpha^N}\left(k + \frac{\varepsilon_N p}{2}\right) \mathbb{E} \left[ \hat{q}_{k+\frac{\varepsilon_N p}{2}} \hat{q}^*_{k-\frac{\varepsilon_N p}{2}} \right] \, dp \, dk, \\
   &=& \frac{\varepsilon_N}{2} \int_{\mathbb{R} \times \mathbb{T}} \hat{\varphi}^*(p) \widehat{\alpha^N}\left(k - \frac{\varepsilon_N p}{2}\right) \mathbb{E} \left[ \hat{q}_{k+\frac{\varepsilon_N p}{2}} \hat{q}^*_{k-\frac{\varepsilon_N p}{2}} \right] \, dp \, dk
\end{eqnarray*}
 and 
 \begin{align*}
    I_{3,2}^N &= - \frac{\varepsilon_N}{2} \int_{\mathbb{R} \times \mathbb{T}} \hat{\varphi}^*(p) \omega_N\left(k + \frac{\varepsilon_N p}{2}\right) \omega_N\left(k - \frac{\varepsilon_N p}{2}\right) \mathbb{E} \left[ \hat{q}_{k+\frac{\varepsilon_N p}{2}} \hat{q}^*_{k-\frac{\varepsilon_N p}{2}} \right] \, dp \, dk.
\end{align*}
Here we use the Lemma \ref{lem:Fourier Transform}\,.

By Lemma \ref{lem:omega_N_properties}~(iii), we obtain
\begin{eqnarray*}
    &&\widehat{\alpha^N}\left(k + \frac{\varepsilon_N p}{2}\right) + \widehat{\alpha^N}\left(k - \frac{\varepsilon_N p}{2}\right) - 2 \omega_N\left(k + \frac{\varepsilon_N p}{2}\right) \omega_N\left(k - \frac{\varepsilon_N p}{2}\right) \\
    &=& \left( \omega_N\left(k + \frac{\varepsilon_N p}{2}\right) - \omega_N\left(k - \frac{\varepsilon_N p}{2}\right) \right)^2 \\
   & \leq& \sup_{k\in\mathbb{T}}\left|\omega_N'(k)\varepsilon_N p\right|^2 
    = O\left((\delta_N N \varepsilon_N p)^2\right).
\end{eqnarray*}
Then Utilizing the relation above yields the final bound
\begin{align*}
    |2 I_3^N| \quad&=\quad |2I_{3,1}^N + 2I_{3,2}^N| \\
    &= \quad\frac{\varepsilon_N}{2} \int_{\mathbb{R} \times \mathbb{T}} |\hat{\varphi}^*(p)| O\left((\delta_N N \varepsilon_N p)^2\right) \mathbb{E} \left[ \left|\hat{q}_{k+\frac{\varepsilon_N p}{2}} \hat{q}^*_{k-\frac{\varepsilon_N p}{2}}\right| \right] \, dp \, dk \\
    &\lesssim \quad\delta_N^2 N^3 \varepsilon_N^3 \left( \frac{1}{N} \mathbb{E}[\mathcal{E}_N] \right) \int_{\mathbb{R}} |\hat{\varphi}^*(p)| p^2 \, dp.
\end{align*}
Since $\hat\varphi \in \mathcal{S}(\mathbb{R})$, the integral $\int_{\mathbb{R}} |\hat{\varphi}^*(p)| p^2 \, dp < \infty$. It follows that  
$$I_3^N = C_\varphi O(\delta_N^2 N^3 \varepsilon_N^3).$$
 Finally, we consider $I_4^N$. By Lemma \ref{lem:alpha_properties}~(iii) and \eqref{phi_exp——decay},  
\begin{eqnarray*}
    |4I_4^N|&=&\left|\varepsilon_N \sum_{|x|> N} \varphi(\varepsilon_N x) \mathbb{E} \Bigg[ \sum_{y=-N}^{N}\alpha_{x-y}^Nq_y^2 \Bigg] \right|
 \\  &\lesssim& C_\varphi\varepsilon_N e^{-b_0N\varepsilon_N}\sum_{x\in\mathbb{Z}}
    |\alpha_{x}^N|\sum_{y=-N}^{N}\mathbb{E}[q_y^2]
    \lesssim  C_\varphi N \varepsilon_N e^{-b_0N\varepsilon_N}.
\end{eqnarray*}

Combining the estimates above, we conclude that
\begin{eqnarray*}
    &&\sup_{t\geq 0}\left|\varepsilon_N \sum_{x\in\mathbb{Z}} \varphi(\varepsilon_N x) \mathbb{E} \left[  e_x(a_Nt) -\frac{1}{2}\left| \psi_x(a_Nt) \right|^2\right]\right| \\
    &=& I_1^N + I_2^N + I_3^N + I_4^N\\
    &=& C_\varphi O(N \varepsilon_N e^{-b_0N\varepsilon_N}) + C_\varphi O(\delta_N^2 N^3 \varepsilon_N^3).
\end{eqnarray*}
\end{proof}
By \ref{H2}, this error vanishes. We thus approximate the energy via the wave function and focus on its dynamics hereafter.

\section{ Wigner Transform}
The preceding proposition reduces the evaluation of $$\varepsilon_N \sum_{x=-N}^{N} \varphi(\varepsilon_N x) \mathbb{E} \left[  e_x(a_Nt) \right]$$ to estimating $$\varepsilon_N \sum_{x\in\mathbb{Z}} \varphi(\varepsilon_N x) \mathbb{E} \left[  \frac{1}{2}| \psi_x(a_Nt) |^2\right].$$  Writing \[|\psi_x(a_Nt) |^2=\psi_x(a_Nt) \psi_x^*(a_Nt)\]  and applying Lemma \ref{lem:Fourier Transform}, we obtain
\begin{eqnarray}
&&\varepsilon_N \sum_{x\in\mathbb{Z}} \varphi(\varepsilon_N x) \mathbb{E} \left[  \frac{1}{2}\left| \psi_x(a_Nt) \right|^2\right] \notag\\
&=&\frac{\varepsilon_N}{2}\int_{\mathbb{R} \times \mathbb{T}}
\hat{\varphi}^*(p)
    \mathbb{E} \left[\hat{\psi}(a_{N}t,k+\frac{\varepsilon_N p}{2})\hat{\psi}^{*}(a_{N}t,k-\frac{\varepsilon_N p}{2})\right] \, dp \, dk.
\end{eqnarray}

Now define the Wigner transform as
\begin{equation}
    \widehat W_{N}(t,p,k):=\frac{\varepsilon_N}{2}
    \mathbb{E} \left[\hat{\psi}(a_{N}t,k+\frac{\varepsilon_N p}{2})\hat{\psi}^{*}(a_{N}t,k-\frac{\varepsilon_N p}{2})\right].
\end{equation}
and the action on  $\varphi \in C_c^\infty(\mathbb{R}\times \mathbb{T})$ of  the Wigner transform is  
\begin{align*}
\langle W_{N}(t),\varphi \rangle
:&=\langle \widehat W_{N}(t),\hat\varphi \rangle_{L^2(\mathbb{R}\times \mathbb{T})}\notag
\\& = \frac{\varepsilon_N}{2} \int_{\mathbb{R} \times \mathbb{T}} \hat{\varphi}^*(p, k)\mathbb{E} \left[ \hat{\psi}^* \left( a_Nt, k - \frac{\varepsilon_N p}{2} \right) \hat{\psi}^{} \left( a_Nt, k + \frac{\varepsilon_N p}{2} \right) \right]  \, dp \, dk.
\end{align*}
Here, $\hat{\varphi}$ denotes the Fourier transform of $\varphi$ with respect to the $p$-variable. Now for  $\varphi(y,k) \equiv \varphi(y)\in \mathcal{C}$\,, 
\begin{equation} \label{eq:wigner}
    \langle W_{N}(t),\varphi \rangle=  \varepsilon_N \sum_{x \in \mathbb{Z}} \varphi(\varepsilon_N x) \mathbb{E} \left[  \frac{1}{2}\left| \psi_x(a_Nt) \right|^2\right].
\end{equation}
In the following we give a simple estimate.
\begin{lemma}\label{lem:WNsimple estimates}
For all $t \ge 0$, $p \in \mathbb{R}$ and $k \in \mathbb{T}$,  
\begin{equation}
        \sup_{t,p}\int_\mathbb{T} \left| \widehat W_{N}(t,p,k)\right|dk=O(N\varepsilon_N)\,.
    \end{equation}
\end{lemma}

\begin{proof}
By the triangle inequality and the translation invariance of the integral over $\mathbb{T}$, we bound the integral as
\begin{eqnarray*}
    \sup_{t, p} \int_{\mathbb{T}} |\widehat{W}_N(t, p, k)| \, dk 
     &\le& \frac{\varepsilon_N}{4} \sup_{t, p} \mathbb{E} \bigg[ \int_{\mathbb{T}} \Big| \hat{\psi}\Big(a_N t, k + \frac{\varepsilon_N p}{2}\Big) \Big|^2 \, dk 
     \\ &&{}+ \int_{\mathbb{T}} \Big| \hat{\psi}^*\Big(a_N t, k - \frac{\varepsilon_N p}{2}\Big) \Big|^2 \, dk \bigg] \\
    &=& \frac{\varepsilon_N}{2} \sup_{t} \mathbb{E} \bigg[ \int_{\mathbb{T}} |\hat{\psi}(a_N t, k)|^2 \, dk \bigg].
\end{eqnarray*}
Applying the Plancherel theorem and \eqref{leq:wave function energy}, we have
$$
\sup_{t} \mathbb{E} \bigg[ \int_{\mathbb{T}} |\hat{\psi}(a_N t, k)|^2 \, dk \bigg] = \sup_{t}\mathbb{E} \bigg[ \sum_{x \in \mathbb{Z}} |\psi_x(a_N t)|^2 \bigg] \lesssim N.
$$
Therefore, we conclude that $\sup_{t, p} \int_{\mathbb{T}} |\widehat{W}_N(t, p, k)| \, dk = O(N\varepsilon_N)$.
\end{proof}

\subsection{Infinitesimal generator and Wigner equations} 
In order to compute $\partial_{t}\widehat W_{N}(t,p,k)$, we consider the infinitesimal generator $\mathcal{G}$ associated with the system. For any sufficiently smooth function $f(q,p)$, taking expectations on both sides of \eqref{eq:itof}, we obtain the infinitesimal generator 
\begin{align}
\mathcal G f
&= \sum_{x=-N}^N \left(
p_x \cdot \frac{\partial f}{\partial q_x}
-(\alpha^N*q)_x \cdot \frac{\partial f}{\partial p_x}
\right)
+\sum_{x=-N}^N \eta_x^N q_x \cdot \frac{\partial f}{\partial p_x} \notag \\
&\quad
+\gamma\sum_{x=-N}^N
\left(
Jp_x \cdot \frac{\partial^2 f}{\partial p_x^2}Jp_x
-p_x \cdot \frac{\partial f}{\partial p_x}
\right).
\end{align}
Now introduce the Wigner-type transforms
\begin{align}
\widehat W_N(t,p,k)
&:=\frac{\varepsilon_N}{2}
\mathbb E\left[
\widehat{\psi}\left(a_Nt,k+\frac{\varepsilon_N p}{2}\right)
\widehat{\psi}^{\,*}\left(a_Nt,k-\frac{\varepsilon_N p}{2}\right)
\right], \label{def:WN}\\
\widehat Y_N(t,p,k)
&:=\frac{\varepsilon_N}{2}
\mathbb E\left[
\widehat{\psi}\left(a_Nt,k+\frac{\varepsilon_N p}{2}\right)
\widehat{\psi}\left(a_Nt,-k+\frac{\varepsilon_N p}{2}\right)
\right], \label{def:YN}
\end{align}
and set
\begin{align}
\widehat Y_{N,-}(t,p,k)
&:=\widehat Y_N^*(t,-p,k),\\
\widehat W_{N,-}(t,p,k)
&:=\widehat W_N(t,p,-k),\\
\widehat U_N(t,p,k)
&:=\frac12\left(\widehat Y_N(t,p,k)+\widehat Y_{N,-}(t,p,k)\right), \label{def:UN}\\
\widehat U_{N,-}(t,p,k)
&:=\frac{1}{2i}\left(\widehat Y_N(t,p,k)-\widehat Y_{N,-}(t,p,k)\right). \label{def:UNminus}
\end{align}
For later purposes, define
\begin{align}
\Delta_N\omega
&:=\omega_N\left(k+\frac{\varepsilon_N p}{2}\right)
-\omega_N\left(k-\frac{\varepsilon_N p}{2}\right), \label{Def:Delta N omega}\\
\bar\omega_N
&:=\omega_N\left(k+\frac{\varepsilon_N p}{2}\right)
+\omega_N\left(k-\frac{\varepsilon_N p}{2}\right), \label{Def:bar omega N}\\
\mathcal Lf(k)
&:=\langle f,1\rangle_{L^2(\mathbb T)}-f(k)
=\int_{\mathbb T}f(k')\,dk'-f(k). \label{def:花体L}
\end{align}
\par Differentiating $\widehat W_N$, $\widehat W_{N,-}$, $\widehat U_N$, and $\widehat U_{N,-}$ with respect to \(t\), we obtain the closed system
\begin{align}
\partial_t\widehat W_N
&=-ia_N\Delta_N\omega\,\widehat W_N
+a_N\gamma\mathcal L(\widehat W_N-\widehat U_N)
+\widehat{\mathcal R}^{(1)}_N, \label{piandao WN}\\
\partial_t\widehat W_{N,-}
&=ia_N\Delta_N\omega\,\widehat W_{N,-}
+a_N\gamma\mathcal L(\widehat W_{N,-}-\widehat U_N)
+\widehat{\mathcal R}^{(2)}_N, \label{piandao WN-}\\
\partial_t\widehat U_N
&=a_N\bar\omega_N\,\widehat U_{N,-}
+a_N\gamma\mathcal L\left(
\widehat U_N-\frac12(\widehat W_N+\widehat W_{N,-})
\right)
+\widehat{\mathcal R}^{(3)}_N, \label{piandao UN}\\
\partial_t\widehat U_{N,-}
&=-a_N\bar\omega_N\,\widehat U_N
-a_N\gamma\widehat U_{N,-}
+\widehat{\mathcal R}^{(4)}_N, \label{piandao UN-}
\end{align}
where the remainder term associated with $\eta^N$ reads
\begin{align}
\widehat{\mathcal R}^{(1)}_N
&=\frac{i a_N\varepsilon_N}{2}
\mathbb E\left[
-\widehat{\psi}_{k+\frac{\varepsilon_N p}{2}}
\widehat{\eta^N q}^{\,*}\left(k-\frac{\varepsilon_N p}{2}\right)
+\widehat{\psi}_{k-\frac{\varepsilon_N p}{2}}^*
\widehat{\eta^N q}\left(k+\frac{\varepsilon_N p}{2}\right)
\right],
\label{eq:R1}\\
\widehat{\mathcal R}^{(2)}_N
&=\frac{i a_N\varepsilon_N}{2}
\mathbb E\left[
-\widehat{\psi}_{-k+\frac{\varepsilon_N p}{2}}
\widehat{\eta^N q}^{\,*}\left(-k-\frac{\varepsilon_N p}{2}\right)
+\widehat{\psi}_{-k-\frac{\varepsilon_N p}{2}}^*
\widehat{\eta^N q}\left(-k+\frac{\varepsilon_N p}{2}\right)
\right],
\label{eq:R2}\\
\widehat{\mathcal R}^{(3)}_N
&=\frac{i a_N\varepsilon_N}{4}
\mathbb E\Bigg[
\widehat{\psi}_{-k+\frac{\varepsilon_N p}{2}}
\widehat{\eta^N q}\left(k+\frac{\varepsilon_N p}{2}\right)
+\widehat{\psi}_{k+\frac{\varepsilon_N p}{2}}
\widehat{\eta^N q}\left(-k+\frac{\varepsilon_N p}{2}\right) \notag\\
&\qquad\qquad
-\widehat{\psi}_{-k-\frac{\varepsilon_N p}{2}}^*
\widehat{\eta^N q}^{\,*}\left(k-\frac{\varepsilon_N p}{2}\right)
-\widehat{\psi}_{k-\frac{\varepsilon_N p}{2}}^*
\widehat{\eta^N q}^{\,*}\left(-k-\frac{\varepsilon_N p}{2}\right)
\Bigg],
\label{eq:R3}\\
\widehat{\mathcal R}^{(4)}_N
&=\frac{a_N\varepsilon_N}{4}
\mathbb E\Bigg[
\widehat{\psi}_{-k+\frac{\varepsilon_N p}{2}}
\widehat{\eta^N q}\left(k+\frac{\varepsilon_N p}{2}\right)
+\widehat{\psi}_{k+\frac{\varepsilon_N p}{2}}
\widehat{\eta^N q}\left(-k+\frac{\varepsilon_N p}{2}\right) \notag\\
&\qquad\qquad
+\widehat{\psi}_{-k-\frac{\varepsilon_N p}{2}}^*
\widehat{\eta^N q}^{\,*}\left(k-\frac{\varepsilon_N p}{2}\right)
+\widehat{\psi}_{k-\frac{\varepsilon_N p}{2}}^*
\widehat{\eta^N q}^{\,*}\left(-k-\frac{\varepsilon_N p}{2}\right)
\Bigg].
\label{eq:R4}
\end{align}

\begin{remark} \label{remark:粗略估计}
Proceeding as in the proof of Lemma \ref{lem:WNsimple estimates}, we can derive rough bounds for quantities with a structure similar to that of $\widehat W_N$. For example,    
\begin{equation}
    \sup_{t\ge0,p\in\mathbb{R}} \left(\frac{1}{a_N}\left\|\widehat{\mathcal R}^{(1)}_N(t,p,\cdot)\right\|_{L^1(\mathbb{T})}
    +
    \left\|\widehat U_N(t,p,\cdot)\right\|_{L^1(\mathbb{T})}\right)
    =O(N\varepsilon_N).
\end{equation}
\end{remark}

For $t>0,\ i\in\{1,2,3,4\}$, define the time-integrated quantities
\begin{align*}
\bar w_N(t,p,k)
&:=\int_0^t\widehat W_N(u,p,k)\,du,
&
\bar w_{N,-}(t,p,k)
&:=\int_0^t\widehat W_{N,-}(u,p,k)\,du,\\
\bar u_N(t,p,k)
&:=\int_0^t\widehat U_N(u,p,k)\,du,
&
\bar u_{N,-}(t,p,k)
&:=\int_0^t\widehat U_{N,-}(u,p,k)\,du,\\
\bar r_N^{(i)}(t,p,k)
&:=\int_0^t\widehat{\mathcal R}_N^{(i)}(u,p,k)\,du\,.
\end{align*}
Integrating \eqref{piandao WN}--\eqref{piandao UN-} over $[0,t]$, we obtain

\begin{align}
&ia_N\Delta_N\omega\,\bar w_N
=\widehat W_N^{(0)}-\widehat W_N
+a_N\gamma\mathcal L(\bar w_N-\bar u_N)
+\bar r_N^{(1)}, \label{piandao wN}\\
&-ia_N\Delta_N\omega\,\bar w_{N,-}
=\widehat W_{N,-}^{(0)}-\widehat W_{N,-}
+a_N\gamma\mathcal L(\bar w_{N,-}-\bar u_N)
+\bar r_N^{(2)}, \label{piandao wN-}\\
&\widehat U_N
=\widehat U_N^{(0)}
+a_N\bar\omega_N\,\bar u_{N,-}
+a_N\gamma\mathcal L\left(
\bar u_N-\frac12(\bar w_N+\bar w_{N,-})
\right)
+\bar r_N^{(3)}, \label{piandao uN}\\
&\widehat U_{N,-}
=\widehat U_{N,-}^{(0)}
-a_N\bar\omega_N\,\bar u_N
-a_N\gamma\bar u_{N,-}
+\bar r_N^{(4)}. \label{piandao uN-}
\end{align}
Here $\widehat W_N^{(0)}(p,k)=\widehat W_N(0,p,k)$, and the same convention is used for $\widehat W_{N,-}^{(0)}$, $\widehat U_N^{(0)}$, and $\widehat U_{N,-}^{(0)}$.

\section{Some key estimates}
We intend to pass the limit $N\rightarrow\infty$ in equation \eqref{piandao wN}.  For this we introduce 
\begin{align} \label{def:D_N}
    D_{N}(p,k):=\gamma+i\Delta_N\omega.
\end{align}
Then, by Lemma \ref{lem:omega_N_properties}~(iii),  
\begin{align} \label{le:DNbounded}
    1\lesssim\sup_{p,k}\left|\frac{\gamma}{D_N(p,k)}\right|\le1.
\end{align}
Multiplying both sides of \eqref{piandao wN} by $\frac{\gamma}{D_{N}}$, and integrating with respect to $k$, we obtain
\begin{align} \label{eq:wn}
 &a_N\gamma\left(1-\int_{\mathbb T}\frac{\gamma}{D_{N}(p,k)}\,dk\right)
 \int_{\mathbb T}\bar w_N(t,p,k)\,dk
\notag \\=&\; \gamma\int_{T}\frac{\widehat{W}_N^{(0)}(p,k)-\widehat{W}_N(t,p,k)}{D_{N}(p,k)}\,dk \notag \\
&-a_N\gamma\int_{\mathbb T}\bar u_{N}(t,p,k) \mathcal L\left(\frac{\gamma}{D_N}\right)(p,k)\,dk
+\gamma\int_{\mathbb T}
\frac{\bar r_{N}^{(1)}(t,p,k)}
{D_{N}(p,k)}\,dk.
\end{align}

\par We  proceed to estimate each term. 

\par To perform the Taylor expansion, set $P_N=(\delta_N N\varepsilon_N)^{-\frac{1}{2}}.$ By \ref{H2},
$P_N \to \infty, \ \delta_N N\varepsilon_N P_N \to 0$.  Then there exists an integer $m_1 \ge 3$ such that
\begin{align*}
    \frac{P_N^{-m_1}}{\delta_N\varepsilon_N} \lesssim \delta_N^2 N^3 \varepsilon_N^2 \quad \text{and} \quad N\varepsilon_N P_N^{-m_1} \to 0.
\end{align*}
By Lemma \ref{lem:omega_N_properties}~(iii), for any $m\in\mathbb{N}$, 
$$\sup_{k\in\mathbb{T}}\left| \frac{d^m}{dk^m} \omega_N(k) \right| = C_mO\left( (\delta_N N)^m \right)\,.$$
 In the regime $|p|\le P_N$, the condition $\delta_N N\varepsilon_N P_N \to 0$ justifies the following  expansion of~$\Delta_N\omega$, 
\begin{align}
    \Delta_N\omega &= \omega_{N}\left(k +\frac{\varepsilon_N p}{2}\right) - \omega_{N}\left(k-\frac{\varepsilon_N p}{2}\right) 
    = \varepsilon_N p\omega_N^{\prime}(k) + O\left((\delta_N N\varepsilon_N |p|)^3\right).
\end{align}
Furthermore,  for $|p|\le P_N$,  
\begin{align} \label{eq:DNtalor}
    \frac{\gamma}{D_N(p,k)} &= \left( 1  + \frac{i \varepsilon_N p \omega'_N(k)}{\gamma}+O\left((\delta_N N\varepsilon_N |p|)^3\right)\right)^{-1} \notag  \\
    &= 1  - \frac{i \varepsilon_N p \omega'_N(k)}{\gamma}  - \frac{\varepsilon_N^2 p^2 (\omega'_N(k))^2}{\gamma^2} + O\left((\delta_N N\varepsilon_N |p|)^3\right).
\end{align}
The coefficient on the left-hand side of \eqref{eq:wn} determines the effective diffusion term in the limit. Accordingly, define
\begin{align}
    d_N(p) := a_N\gamma\left(1-\int_{\mathbb T}\frac{\gamma}{D_{N}(p,k)}\,dk\right).
\end{align}
The next proposition gives the limiting form of $d_N(p)$.
\begin{proposition} \label{prop:main dNp}
For any $\varphi \in \mathcal{S}(\mathbb{R}),m \in \mathbb{N}_+,$ 
\begin{eqnarray} \label{guji1}
   &&\sup_{t\geq0}\left|  \int_{\mathbb{R}\times\mathbb T}\hat{\varphi}^*(p)\widehat{W}_N(t,p,k)\left( d_N(p)-\frac{\hat c}{\gamma}p^2\right)  \,dp\,dk \right|
   \notag \\&=& C_\varphi O\left(N\varepsilon_N(\delta_N N)^{-m}\right)+C_\varphi O(\delta_N^2 N^3 \varepsilon_N^2).
\end{eqnarray}
\end{proposition}

\begin{proof}

Since $\omega_N$ is  even, its derivative $\omega_N'$ is odd, which implies $\int_{\mathbb{T}} \omega_N'(k) \, dk = 0$. So for $|p| \leq P_N,$ integrating both sides of \eqref{eq:DNtalor} with respect to $k$,
\begin{align*}
    &1 - \int_{\mathbb{T}} \frac{\gamma}{D_N(p,k)} \, dk 
     =  \frac{\varepsilon_N^2 p^2}{\gamma^2} \int_{\mathbb{T}} |\omega_N'(k)|^2 \, dk 
     +O\left((\delta_N N\varepsilon_N |p|)^3\right).
\end{align*}

Recall that $a_N$\,, defined in \eqref{def:aN}\,, satisfies $a_N \asymp (\delta_N N \varepsilon_N^2)^{-1}$, by Lemma~\ref{lem:aN limit},  
\begin{eqnarray*}
    \left| a_N\varepsilon_N^2\int_{\mathbb{T}} |\omega_N'(k)|^2 \, dk - \hat c \right| = O\left((\delta_N N)^{-m}\right).
\end{eqnarray*}
Consequently, for $|p| \leq P_N$, we have
\begin{eqnarray*}
    &&\left| d_N(p) - \frac{\hat c}{\gamma} p^2 \right| \\
    &=& \left| \frac{a_N \varepsilon_N^2 p^2}{\gamma} \int_{\mathbb{T}} |\omega_N'(k)|^2 \, dk - \frac{\hat c}{\gamma} p^2 + O\left(a_N(\delta_N N\varepsilon_N |p|)^3\right) \right| \\ 
    &=& O\left((\delta_N N)^{-m}|p|^2\right) + O\left(\delta_N^2N^2\varepsilon_N |p|^3\right).
\end{eqnarray*}
By Lemma \ref{lem:WNsimple estimates}, combining this with the previous estimate yields
\begin{eqnarray*}
    && \int_{\{|p|\leq P_N\} }\int_{ \mathbb{T}} \Big|\hat{\varphi}^*(p) \widehat{W}_N(t,p,k)\Big|\cdot \Big| d_N(p) -  \frac{\hat c}{\gamma} p^2 \Big| \, dp \, dk \\
    &\lesssim& \int_{\mathbb{R}} |\hat{\varphi}^*(p)| (1+|p|^2+|p|^3) \, dp \cdot \left(O\left(N\varepsilon_N(\delta_N N)^{-m}\right)+ \delta_N^2 N^3 \varepsilon_N^2 \right).
\end{eqnarray*}
Since $\varphi \in \mathcal{S}(\mathbb{R})$, the integral $\int_{\mathbb{R}} |\hat{\varphi}^*(p)| (1+|p|^2+|p|^3) \, dp$ is finite.  
\par
The tail regime $|p|>P_N$ will subsequently be controlled utilizing the boundedness of $d_N$ and the rapid decay of the Schwartz test functions. In fact,  
by \eqref{def:D_N},  $|d_N(p)|\leq 2\gamma a_N$\,, and 
\begin{align*}
    \left| d_N(p)-\frac{\hat c}{\gamma}p^2\right| \lesssim a_N+p^2\,. 
\end{align*}
 Then  \begin{eqnarray*}
    &&\int_{\{|p|\geq P_N\} }\int_{ \mathbb{T}} \hat{\varphi}^*(p) \widehat{W}_N(t,p,k) \left| d_N(p) - \frac{\hat c}{\gamma} p^2 \right| \, dp \, dk \notag \\
    &\lesssim& N\varepsilon_N \int_{\{|p|\geq P_N\} } |\hat{\varphi}^*(p)|(a_N+p^2) \, dp \\
   & \lesssim& C_\varphi \frac{P_N^{-m_1}}{\delta_N\varepsilon_N} 
    \lesssim C_\varphi\delta_N^2 N^3 \varepsilon_N^2.
\end{eqnarray*}
Combining the estimates above completes the proof.
\end{proof}


Next, we identify 
the limit of the first term on the right-hand side of \eqref{eq:wn}.
\begin{lemma} 
For any $\varphi \in \mathcal{C}$,  
\begin{eqnarray}
    &&\sup_{t\ge 0} \left|\int_{\mathbb{R} \times \mathbb{T}} \widehat{\varphi}^*(p) 
  \left( \widehat{W}_N^{(0)}(p,k)-\widehat{W}_N(t,p,k) \right) 
  \left( \frac{\gamma}{D_N(p,k)} - 1 \right) \, dp \, dk\right| 
  \notag \\
  &=& C_\varphi O(\delta_N N^2\varepsilon_N^2) + C_\varphi O(N\varepsilon_N P_N^{-m_1}).
\end{eqnarray}
\end{lemma}

\begin{proof}
    The discussion  is similar to that for Proposition \ref{prop:main dNp}. By \eqref{eq:DNtalor} and Lemma~\ref{lem:WNsimple estimates}
\begin{align*}
    &\sup_{t\ge 0}\left| \int_{\{|p| \le P_N\}} \int_{\mathbb{T}} 
    \hat{\varphi}^*(p)  
    \left( \widehat{W}_N^{(0)}(p,k)-\widehat{W}_N(t,p,k) \right) 
    \bigg( \frac{\gamma}{D_N(p,k)} - 1 \bigg) \, dk \, dp \right| \\
     \le\quad& \sup_{p \in \mathbb{R}} \int_{\mathbb{T}} 
     \left( |\widehat{W}_N^{(0)}(p,k)|+\sup_{t\ge 0}|\widehat{W}_N(t,p,k)|\right) \, dk
    \int_{\mathbb{R}} |\hat{\varphi}^*(p)| \delta_N N\varepsilon_N |p|  \, dp \\
    \lesssim\quad& \bigg( \int_{\mathbb{R}} |\hat{\varphi}^*(p)|(1+|p|) \, dp \bigg)
      \delta_N N^2\varepsilon_N^2 
    = C_\varphi O(\delta_N N^2\varepsilon_N^2)
\end{align*}
which vanishes  by \ref{H2}.  

For the tail part, using \eqref{le:DNbounded} and Lemma \ref{lem:WNsimple estimates}, we get
\begin{align*}
    &\left| \int_{\{|p| > P_N\}} \int_{\mathbb{T}} 
    \hat{\varphi}^*(p)  
    \left( \widehat{W}_N^{(0)}(p,k)-\widehat{W}_N(t,p,k) \right) 
    \bigg( \frac{\gamma}{D_N(p,k)} - 1 \bigg) \, dk \, dp \right| \\
     \le\quad& 4N\varepsilon_N \int_{\{|p| > P_N\}} |\hat{\varphi}^*(p)| \, dp
    \lesssim C_\varphi O( N\varepsilon_N P_N^{-m_1} ).
\end{align*}
This completes the proof.
\end{proof}
Therefore, by \ref{H2},
\begin{align} \label{guji2}
    \left|
\left\langle
\gamma \int_{\mathbb{T}} \frac{\widehat{W}_N^{(0)}(p,k)-\widehat{W}_N(t,p,k)}{D_N(p,k)} \, dk,
\hat{\varphi}
\right\rangle_{L^2(\mathbb{R})}
-
\left\langle \widehat{W}_N^{(0)}-\widehat{W}_N(t), \hat{\varphi} \right\rangle_{L^2(\mathbb{R}\times \mathbb{T})}
\right|
\to 0.
\end{align}
By  Proposition \ref{thm:wavebijin},
$$
\left| \left\langle \widehat{W}_N^{(0)}, \hat{\varphi} \right\rangle_{L^2(\mathbb{R}\times \mathbb{T})}- I_N(0,\varphi) 
\right|
\to 0\,.
$$
However,   by \ref{H3},    
$$
I_N(0,\varphi) \to \langle \mu_0, \varphi\rangle
$$
just holds for $\varphi \in C_c^\infty(\mathbb{R})$\,. 
The following lemma extends this convergence from $C_c^\infty(\mathbb R)$ to $\mathcal C$.

\begin{lemma} \label{lem:t=0 H2 C_function}
For $\varphi \in \mathcal{C}$,
\begin{equation}
    \lim_{\substack{ N\to \infty}} \varepsilon_N \sum_{x \in \mathbb{Z}} \varphi( \varepsilon_N x) \mathbb{E}\left[ e_x \left( 0 \right)\right] = \int_{\mathbb{R}} \varphi(u) \mu_0(du).
\end{equation}
\end{lemma}

\begin{proof}
When $\varphi \in C_c^\infty(\mathbb{R}),$ the conclusion follows under \ref{H3}. For any $\varphi \in \mathcal{C}$, it suffices to show that $I_N(0, \varphi) \to \langle \mu_0, \varphi \rangle$ as $N \to \infty$.
Let $\chi\in C_c^\infty(\mathbb R)$ be a smooth cutoff function such that
$0\le\chi\le1$, $\chi(p)=1$ for $|p|\le1$, $\chi(p)=0$ for $|p|\ge2$.
For any $m \ge 1$, we define
$$
\varphi_m(x) = \varphi(x) \chi\left(\frac{x}{m}\right).
$$
Then $\varphi_m \in C_c^\infty(\mathbb{R})$, and $\varphi_m \to \varphi$ as $m \to \infty$. By the triangle inequality, we have
\begin{align} \label{*1}
\left| I_N(0, \varphi) - \langle \mu_0, \varphi \rangle \right| \le &\left| I_N(0, \varphi) - I_N(0, \varphi_m) \right| \nonumber \\
&+ \left| I_N(0, \varphi_m) - \langle \mu_0, \varphi_m \rangle \right| + \left| \langle \mu_0, \varphi_m \rangle - \langle \mu_0, \varphi \rangle \right|.
\end{align}

By \ref{H3}, $I_N(0, \varphi_m) \to \langle \mu_0, \varphi_m \rangle$ as $N \to \infty$. Furthermore, \ref{H3} also implies $\varphi \in L^1(\mu_0)$, which yields the uniform bound $\left| \varphi_m - \varphi \right| \le 2|\varphi| \in L^1(\mu_0)$. Applying the dominated convergence theorem, we obtain $\langle \mu_0, \varphi_m \rangle \to \langle \mu_0, \varphi \rangle$ as $m \to \infty$.

For the first term, using the support property of the cutoff function, we obtain
\begin{align*}
\left| I_N(0, \varphi) - I_N(0, \varphi_m) \right| \quad&=\quad \left| \varepsilon_N \sum_{|x| \le N} \varphi(\varepsilon_N x) \left( 1 - \chi\left(\frac{\varepsilon_N x}{m}\right) \right) \mathbb{E}[e_x(0)] \right| \\
&\le\quad \varepsilon_N \sum_{m/\varepsilon_N \le | x| \le N} \left| \varphi(\varepsilon_N x) \right| \mathbb{E}[e_x(0)] \\
&\lesssim\quad C_\varphi\varepsilon_N \sum_{m/\varepsilon_N \le | x| \le N} e^{-b_1\varepsilon_N|x|} e^{-(b_0-b_1)m} \mathbb{E}[e_x(0)] \\
&\le\quad C_\varphi e^{-(b_0-b_1)m} \sup_N \varepsilon_N \sum_{x=-N}^N e^{-b_1\varepsilon_N|x|} \mathbb{E}[e_x(0)] \\
&\lesssim\quad C_\varphi e^{-(b_0-b_1)m},
\end{align*}
where $b_1$ is the constant appearing in \ref{H4}.
Therefore, taking $N \to \infty$ first and then $m \to \infty$ in \eqref{*1}, the conclusion holds.
\end{proof}

Then by Lemma \ref{lem:t=0 H2 C_function} and \ref{H2},   for $\varphi\in \mathcal{C}$
\begin{align} \label{guji3}
    \left\langle
\gamma \int_{\mathbb{T}} \frac{\widehat{W}_N^{(0)}(p,k)}
{D_N(p,k)} \, dk,
\hat{\varphi} 
\right\rangle
\to  \langle \mu_0, \varphi\rangle.
\end{align}
\par Next, we proceed to estimate the remainder terms $R_N^{(i)}, i\in\{1,2,3,4\}$, arising from the free boundary condition. It is worth noting that without the presence of the test function, the asymptotic order of these terms would blow up. Therefore, the test function plays a crucial role in the analysis, and its regularizing effect must be exploited through integration by parts. To facilitate this procedure, we first establish the necessary estimates for the higher-order derivatives in the following lemma.
\begin{lemma} \label{lem:DN higher-order derivative}
For every $n\ge1$, there exists $C_n>0$, such that
\begin{enumerate}[label=\textup{(\roman*)}, align=left]
    \item \begin{equation*}
        \sup_{\substack{|p|\le2P_N \\ k\in\mathbb{T}}}\left|\frac{\gamma}{D_N(p,k)}-1\right|
=O(\delta_N N\varepsilon_N P_N).
    \end{equation*}
    \item \begin{align*}
\sup_{\substack{p\in\mathbb{R} \\ k\in\mathbb{T}}}
\left| \partial_p^n\left( \frac{\gamma}{D_N(p,k)}-1 \right) \right|
= C_nO\left((\delta_N N\varepsilon_N)^n\right).
\end{align*}
   \item
\begin{align*}
\sup_{\substack{p\in\mathbb{R} \\ k\in\mathbb{T}}}
\left| \partial_k^n\left( \frac{\gamma}{D_N(p,k)}-1 \right) \right|
=C_n O((\delta_N N)^{n} ).
\end{align*}
\end{enumerate}
\end{lemma}

\begin{proof}
\par
\noindent(i). 
Proceeding directly from \eqref{def:D_N}, we observe that for $|p| \le 2|P_N|$,
\begin{eqnarray*}
|\Delta_N\omega(p,k)|&=&\left|\omega_{N}\left(k +\frac{\varepsilon_N p}{2}\right)-\omega_{N}\left(k -\frac{\varepsilon_N p}{2}\right)\right|
\le \sup_k|\omega_N'(k)|\varepsilon_N |p| \\
&\lesssim &\delta_N N\varepsilon_N P_N.
\end{eqnarray*}
Therefore,
\begin{align}
&\left|\frac{\gamma}{D_N(p,k)}-1\right|
=\left|\frac{-i\Delta_N\omega(p,k)}{\gamma+i\Delta_N\omega(p,k)}\right|
\leq \frac{1}{\gamma}\left|i\Delta_N\omega(p,k)\right|
=O(\delta_N N\varepsilon_N P_N).
\end{align}
\par
\noindent(ii). 
For every $n\ge1$, by Lemma \ref{lem:omega_N_properties}~(iii),
\begin{align*}
\partial_p^n\Delta_N\omega(p,k)
\quad=\quad&
\left(\frac{\varepsilon_N}{2}\right)^n
\left[
\omega_N^{(n)}\left(k+\frac{\varepsilon_N p}{2}\right)
-
(-1)^n\omega_N^{(n)}\left(k-\frac{\varepsilon_N p}{2}\right)\right] \\
=\quad&C_nO\left((\delta_N N\varepsilon_N)^n\right).
\end{align*}
So
\begin{align*}
\partial_p^nD_N(p,k)=C_nO\left((\delta_N N\varepsilon_N)^n\right).
\end{align*}
Recall that $|D_N(p,k)| \ge \gamma.$ Let $F_N(p,k) = \frac{\gamma}{D_N(p,k)}$. Then $F_N(p,k) D_N(p,k) = \gamma$ and $|F_N(p,k)| \le 1$. 

Taking the $n$-th derivative with respect to $p$ on both sides yields
$$
\sum_{j=0}^{n} \binom{n}{j} \partial_p^j F_N(p,k) \partial_p^{n-j} D_N(p,k) = 0.
$$
We proceed by induction. For $n=1$, we have
$$
\partial_p F_N(p,k) D_N(p,k) + F_N(p,k) \partial_p D_N(p,k) = 0.
$$
Thus, the basic case holds since
\begin{align*}
|\partial_p F_N(p,k)| &= \left| \frac{F_N(p,k) \partial_p D_N(p,k)}{D_N(p,k)} \right| 
\le \frac{1}{\gamma} |\partial_p D_N(p,k)| 
= O(\delta_N N \varepsilon_N).
\end{align*}

Assume the estimate holds for all orders up to $n-1$. By the inductive hypothesis, isolating the $n$-th derivative term gives
\begin{align*}
|\partial_p^n F_N(p,k)| &= \left| \frac{\sum_{j=0}^{n-1} \binom{n}{j} \partial_p^j F_N(p,k) \partial_p^{n-j} D_N(p,k)}{D_N(p,k)} \right| \\
&\lesssim C_n \frac{1}{\gamma} \sum_{j=0}^{n-1} \binom{n}{j} (\delta_N N \varepsilon_N)^j \cdot (\delta_N N \varepsilon_N)^{n-j} 
= C_nO\left( (\delta_N N \varepsilon_N)^n \right).
\end{align*}
\par
\noindent(iii). 
Furthermore, noting that
\begin{align*}
\partial_k^n\Delta_N\omega(p,k)
=\left[\omega_N^{(n)}\left(k+\frac{\varepsilon_N p}{2}\right)
-\omega_N^{(n)}\left(k-\frac{\varepsilon_N p}{2}\right)\right]
=C_n O\left((\delta_NN)^n\right),
\end{align*}
a similar argument yields (iii).
\end{proof}
We now turn to the remainder terms arising from the free boundary condition. Although only $\bar r_{N}^{(1)}$ appears explicitly in \eqref{eq:wn}, the remaining remainder terms arise in the subsequent estimate of the term involving $\bar u_N$. It is therefore convenient to treat all of them simultaneously here.
\begin{lemma} \label{lem:yuxiang}
For any $\varphi \in \mathcal{C},\ i\in\{1,2,3,4\}$
\begin{enumerate}[label=\textup{(\roman*)}, align=left]
    \item \begin{equation} 
    \lim_{N\to \infty}\sup_{t\geq0}\left|    \int_{\mathbb{R}}\int_{\mathbb T}\hat{\varphi}^*(p) \widehat{\mathcal R}_N^{(i)}(t,p,k) \,dk\,dp  \right|= 0, 
    \end{equation}
    \item \begin{equation} \label{guji4}
      \lim_{N\to \infty}\sup_{t\geq0}\left| \int_{\mathbb{R}}\int_{\mathbb T}\hat{\varphi}^*(p) \widehat{\mathcal R}_N^{(i)} (t,p,k)
       \frac{\gamma}{D_N(p,k)}
       \,dk\,dp \right|= 0.
    \end{equation}
    \item \begin{equation}
     \lim_{N\to \infty}\sup_{t\geq0}\left|  \int_{\mathbb{R}}\int_{\mathbb T}\hat{\varphi}^*(p) \widehat{\mathcal R}_N^{(i)} (t,p,k)
       \int_{\mathbb T}\frac{\gamma}{D_N(p,k')}\,dk'
       \,dk\,dp \right|= 0.
    \end{equation}
\end{enumerate}

\end{lemma}

\begin{proof}
According to Remark \ref{remark:no time}, we omit the independent variable $t$ and write 
\begin{equation}
    \widehat{\psi}_{k_1} \widehat{\eta^N q}(k_2) =\sum_{x=-N}^{N}\sum_{x'\in \mathbb{Z}} \eta_{x}^N q_{x} \psi_{x'} e^{-2\pi i k_2 x} e^{-2\pi i k_1x'}.
\end{equation}
  Observe that the terms in equations \eqref{eq:R1}--\eqref{eq:R4} are all finite linear combinations of the elementary cross-moments of the form $\mathbb{E}[\widehat{\psi}_{k_1} \widehat{\eta^N q}(k_2)]$ or their complex conjugates. Moreover, since $\widehat{\eta^N q}^*(k)=\widehat{\eta^N q}(-k)$~(as $\eta_{x}^N q_{x}\in \mathbb{R}$), the complex conjugate terms  are all  subject to the condition $k_1 + k_2 = \varepsilon_N p$. 
  So it suffices to consider a representative term, 
  for instance, we can parameterize the frequencies as $k_1 = k + \frac{\varepsilon_N p}{2}$ and $k_2 = -k + \frac{\varepsilon_N p}{2}$ and  set
\begin{align} \label{eq:R0}
    \widehat{\mathcal{R}}^{(0)}_{N}
    &= a_{N} \varepsilon_N\mathbb{E} \left[ \hat{\psi}_{k+\frac{\varepsilon_N p}{2}} 
    \widehat{\eta^N q}(-k+\frac{\varepsilon_N p}{2})\right].
\end{align}
\noindent(i). 
By Lemma \ref{lem:Fourier Transform},
\begin{align}
    \int_{\mathbb{R}}\int_{\mathbb T}\hat{\varphi}^*(p) \widehat{\mathcal R}_N^{(0)} \,dk\,dp
    &=a_{N} \varepsilon_N\int_{\mathbb{R}}\int_{\mathbb T}\hat{\varphi}^*(p) 
\mathbb{E} \left[ \hat{\psi}_{k+\frac{\varepsilon_N p}{2}} \widehat{\eta^N q}(-k+\frac{\varepsilon_N p}{2})\right]
\,dk\,dp \notag \\
&=a_{N} \varepsilon_N \sum_{x=-N}^{N}{\varphi}\left(\varepsilon_N x\right) \mathbb{E} \left[\eta_{x}^N q_{x} \psi_{x}\right].
\end{align}
By Lemma \ref{lem:eta} and \eqref{def:JN}, we have $\eta_x^N=0$ whenever $|x|\leq N-M\delta_NN$. Hence, $\eta_x^N\neq 0$ only if $x\in J_N$, in which case $|x|\geq N-M\delta_NN$. It follows that
$$e^{-b_0\varepsilon_N|x|} \leq e^{-b_0\varepsilon_N(N-M\delta_N N)} \lesssim e^{-b_0 N\varepsilon_N}\,.$$
We  then obtain
\begin{align}
&a_{N} \varepsilon_N \sum_{x=-N}^{N}{\varphi}\left(\varepsilon_N x\right) \mathbb{E} \left[\eta_{x}^N q_{x} \psi_{x}\right]
\quad\lesssim\quad C_\varphi a_{N} \varepsilon_N \sum_{x=-N}^{N}e^{-b_0\varepsilon_N|x|}\mathbb{E} \left[\eta_{x}^N q_{x} \psi_{x}\right]
 \notag\\
\lesssim\quad& C_\varphi a_{N} N\varepsilon_N e^{-b_0N\varepsilon_N}\frac{\mathbb{E}[\mathcal{E}_N(t)]}{N} .
\end{align}
By \ref{H1} and Lemma \ref{lem:conservation of energy},
\begin{align}
  \left| \int_{\mathbb{R}}\int_{\mathbb T}\hat{\varphi}^*(p) \widehat{\mathcal R}_N^{(0)} \,dk\,dp \right|
   \lesssim C_\varphi a_{N} N\varepsilon_N e^{-b_0N\varepsilon_N}
   \asymp C_\varphi \frac{e^{-b_0N\varepsilon_N}}{\delta_N\varepsilon_N}
\end{align}
which  tends to zero by \ref{H2}.

\noindent(ii). 
By part (i), 
\begin{eqnarray*}
    &&\int_{\mathbb R}\int_{\mathbb T}
\hat{\varphi}^*(p)\widehat{\mathcal R}_N^{(0)}
\frac{\gamma}{D_N(p,k)}\,dk\,dp
 \\&=&
\int_{\mathbb R}\int_{\mathbb T}
\hat{\varphi}^*(p)\widehat{\mathcal R}_N^{(0)}
\left(\frac{\gamma}{D_N(p,k)}-1\right)\,dk\,dp
+C_\varphi O\left(\frac{e^{-b_0N\varepsilon_N}}{\delta_N\varepsilon_N}\right)\\
&=&
I_{5,1}^N+I_{5,2}^N++C_\varphi O\left(\frac{e^{-b_0N\varepsilon_N}}{\delta_N\varepsilon_N}\right)
\end{eqnarray*}
with 
\begin{eqnarray*}
I_{5,1}^N&:=&
\int_{\mathbb R}\int_{\mathbb T}
\hat{\varphi}^*(p)\widehat{\mathcal R}_N^{(0)}
\left(\frac{\gamma}{D_N(p,k)}-1\right)
\left(1-\chi\left(\frac{p}{P_N}\right)\right)\,dk\,dp, \\
I_{5,2}^N
&:=&
\int_{\mathbb R}\int_{\mathbb T}
\hat{\varphi}^*(p)\widehat{\mathcal R}_N^{(0)}
\left(\frac{\gamma}{D_N(p,k)}-1\right)
\chi\left(\frac{p}{P_N}\right)\,dk\,dp .
\end{eqnarray*}
Here $\chi\in C_c^\infty(\mathbb R)$ is a smooth cutoff function such that
$0\le\chi\le1$, $\chi(p)=1$ for $|p|\le1$, $\chi(p)=0$ for $|p|\ge2$, and for any $\ell\ge0,\ \sup_p\left|\frac{d^\ell}{dp^\ell}\chi(p)\right|\le1.$

By Remark \ref{remark:粗略估计}, 
\begin{align*}
    \sup_{t,p}\int_\mathbb{T} \left| \widehat{\mathcal R}_N^{(0)}(t,p,k)\right|dk=O(a_NN\varepsilon_N)=O\left(\frac{1}{\delta_N \varepsilon_N}\right).
\end{align*}
Using \eqref{le:DNbounded}, and the rapid decay of $\hat{\varphi}^*$, we get
\[
|I_{5,1}^N|
\lesssim
\frac{1}{\delta_N \varepsilon_N}
\int_{|p|>P_N}|\hat{\varphi}^*(p)|\,dp
\lesssim
\frac{C_\varphi P_N^{-m_1} }{\delta_N\varepsilon_N}.
\]
By the choice of $P_N$, this term tends to $0$.\par
It remains to estimate $I_{5,2}^N$. Expanding the Fourier transforms gives
\begin{align} \label{1}
\widehat\psi_{k+\frac{\varepsilon_N p}{2}}
\widehat{\eta^Nq}\left(-k+\frac{\varepsilon_N p}{2}\right)
=
\sum_{x' \in\mathbb Z}\sum_{x\in J_N}
\eta_x^Nq_x\psi_{x'}
e^{2\pi ik(x-x')}
e^{-\pi i\varepsilon_N p(x+x')}.
\end{align}
Hence
\[
I_{5,2}^N
=
a_N\varepsilon_N
\sum_{x'\in\mathbb Z}\sum_{x\in J_N}
\mathbb E[\eta_x^Nq_x\psi_{x'}]\,
K_N(x,x'),
\]
where
\[
K_N(x,x')
:=
\int_{\mathbb R}\int_{\mathbb T}
e^{2\pi ik(x-x')}
e^{-\pi i\varepsilon_N p(x+x')}
\hat{\varphi}^*(p)
\left(\frac{\gamma}{D_N(p,k)}-1\right)
\chi\left(\frac{p}{P_N}\right)\,dk\,dp .
\]
Note that for any $\ell \ge 0$, $\frac{d^\ell}{dp^\ell}\chi$ is supported in $|p| \leq 2P_N$, and $P_N\delta_NN\varepsilon_N \to 0$, so we only need to consider the estimates in this region.
\par
We now split the estimate of $K_N(x,x')$ into two cases.
First, assume that
\[
|x+x'|>N-M\delta_N(2N+1).
\]
By Lemma \ref{lem:DN higher-order derivative}~(i) and (ii), applying the generalized Leibniz rule, we bound the $n$-th derivative as 
\begin{eqnarray*}
    &&    \int_{\mathbb{R}}  \Bigg| \partial_p^n \bigg[ \hat{\varphi}^*(p) \bigg( \frac{\gamma}{D_N(p,k)} - 1 \bigg) \chi\Big(\frac{p}{P_N}\Big) \bigg] \Bigg|\, dp  \\
     &=& \int_{\mathbb{R}}\bigg| \sum_{j+k+l=n} \frac{n!}{j!k!l!} \partial_p^j \hat{\varphi}^*(p) \cdot \partial_p^k \bigg( \frac{\gamma}{D_N(p,k)} - 1 \bigg) \cdot \partial_p^l \chi\Big(\frac{p}{P_N}\Big) \bigg| \, dp\\
&\lesssim&  \sum_{ \substack{ j+k+l=n \\ k\neq 0} } \frac{n!}{j!k!l!} \int_{\mathbb{R}}|\partial_p^j \hat{\varphi}^*(p)| \, dp\cdot (C_n \delta_N N \varepsilon_N)^k  \\
&&{}+ \sum_{ \substack{ j+l=n \\ } } \frac{n!}{j!l!} \int_{\mathbb{R}}|\partial_p^j \hat{\varphi}^*(p)| \, dp\cdot (C_n \delta_N N \varepsilon_N P_N)  \\
&\lesssim& C_{\varphi,n} \delta_N N \varepsilon_N P_N.
\end{eqnarray*}
Since the boundary terms vanish at infinity, that is
$$
\lim_{|p| \to \infty} \partial_p^n \bigg[ \hat{\varphi}^*(p) \bigg( \frac{\gamma}{D_N(p,k)} - 1 \bigg) \chi\Big(\frac{p}{P_N}\Big) \bigg] = 0,
$$
integrating by parts $n$ times with respect to $p$ gives
\begin{align*}
    |K_N(x,x')| &\lesssim \frac{1}{\varepsilon_N^n |x+x'|^n} \\
    &\quad \times \left| \int_{\mathbb{R}} e^{-2\pi i \varepsilon_N p(x+x')} \partial_p^n \left[ \hat{\varphi}^*(p) \left( \frac{\gamma}{D_N(p,k)} - 1 \right) \chi\left(\frac{p}{P_N}\right) \right] \, dp \right| \\
    &\lesssim \frac{C_{\varphi,n}}{\varepsilon_N^n |x+x'|^n} \delta_N N \varepsilon_N P_N.
\end{align*}
Consequently, by the Cauchy--Schwarz inequality, for every $n\ge2$,
\begin{eqnarray*}
&&
\sum_{\substack{|x+x'|>N-M\delta_N(2N+1)\\ x\in J_N}}
|\eta_x^Nq_x\psi_{x'}K_N(x,x')| \notag \\
&\lesssim&
\left(\sum_{\substack{|x+x'|>N-M\delta_N(2N+1)\\ x\in J_N}}
|\psi_{x'}|^2\right)^{\frac{1}{2}}
\left(\sum_{\substack{|x+x'|>N-M\delta_N(2N+1)\\ x\in J_N}}
|\eta_x^Nq_xK_N(x,x')|^2\right)^{\frac{1}{2}}
\notag \\
&\lesssim&
C_{\varphi,n}|J_N|^{\frac{1}{2}}
\left(\sum_{x'\in\mathbb Z}|\psi_{x'}|^2\right)^{\frac{1}{2}}
\frac{\delta_NN\varepsilon_N P_N}{\varepsilon_N^n}
\left(\sum_{\substack{|x+x'|>N-M\delta_N(2N+1)\\ x\in J_N}}\frac{|\eta_x^Nq_x|^2}{ |x+x'|^{2n}}\right)^{\frac{1}{2}}
\notag \\
&\lesssim&
C_{\varphi,n}(\delta_NN)^{\frac12}
\left(\sum_{x'\in\mathbb Z}|\psi_{x'}|^2\right)^{\frac12}
\frac{\delta_NN\varepsilon_N P_N}{\varepsilon_N^n}
\left( \sum_{|x+x'|>N-M\delta_N(2N+1)}\frac{1}{ |x+x'|^{2n}} \right)^{\frac12}
\\ &&\times
\left( \sum_{x\in J_N}|\eta_x^Nq_x|^2 \right)^{\frac12}. 
\end{eqnarray*}

Notice the fact that  
\begin{eqnarray*}
    &&\left( \sum_{|x+x'|>N-M\delta_N(2N+1)}\frac{1}{ |x+x'|^{2n}} \right)^{\frac12}
    \lesssim C_n \left( \sum_{|y|>N}\frac{1}{ |y|^{2n}} \right)^{\frac12} \\
    &\leq& C_n \left( \sum_{k=1}^{\infty}  \int_{N+k-1}^{N+k}\frac{1}{ |y|^{2n}}
    \, dy \right)^{\frac12}
    =C_n \int_N^\infty \frac{1}{ |y|^{2n}} \, dy
    =C_n N^{\frac12-n}\,, 
\end{eqnarray*}
combining this estimate with \ref{H1}, Lemma \ref{lem:conservation of energy} and Lemma \ref{lem:eta}, we then obtain
\[
\mathbb E\left[
\sum_{\substack{|x+x'|>N-M\delta_N(2N+1)\\ x\in J_N}}
|\eta_x^Nq_x\psi_{x'}K_N(x,x')|
\right]
=
C_{\varphi,n}O\left(\delta_N^{\frac{3}{2}}N^{2}P_N(N\varepsilon_N)^{-n}\right).
\]
Now for 
\[
|x+x'|\le N-M\delta_N(2N+1)\,,
\]
since $x\in J_N$, we have
\[
|x|\ge N-M\delta_N(2N+1).
\]
Hence
\[
|x-x'|
\ge
2|x|-|x+x'|
\ge
N-M\delta_N(2N+1).
\]
For any $m \ge 2$, combining the estimate in Lemma \ref{lem:DN higher-order derivative}~(iii) with integration by parts with respect to $k$ yields
\begin{align*}
|K_N(x,x')| &\lesssim \frac{1}{|x-x'|^m} \int_{\{|p| \le 2P_N\}} |\hat{\varphi}^*(p)| \int_{\mathbb{T}} \bigg| e^{2\pi i k(x-x')} \partial_k^m \bigg( \frac{\gamma}{D_N(p,k)} - 1 \bigg) \bigg| \, dk \, dp \\
&\lesssim \frac{1}{|x-x'|^m} \int_{\{|p| \le 2P_N\}} |\hat{\varphi}^*(p)| \, dp \cdot \delta_N^{m} N^{m}  \\
&\lesssim \frac{C_{\varphi,m} }{|x-x'|^m} (\delta_N N)^m.
\end{align*}
By a similar argument, we have
\[
\begin{aligned}
&\mathbb E\left[
\sum_{\substack{|x+x'|\le N-M\delta_N(2N+1)\\ x\in J_N}}
|\eta_x^Nq_x\psi_{x'}K_N(x,x')|
\right]  \\
\lesssim\quad& 
C_{\varphi,m} (\delta_NN)^{\frac{1}{2}}
\mathbb E\left[\sum_{x'\in\mathbb Z}|\psi_{x'}|^2\right]^{\frac{1}{2}}
(\delta_N N)^m N^{\frac12-m}
\mathbb E\left[\sum_{x\in J_N}|\eta_x^Nq_x|^2\right]^{\frac{1}{2}} \\
\lesssim\quad&
C_{\varphi,m}  \delta_N^{m+\frac12}N^{2}.
\end{aligned}
\]
Combining the above, for every $m\geq 2,\ n\geq2,$ we obtain
\begin{align*}
I_{5,2}^N&=
a_N\varepsilon_N
\sum_{x'\in\mathbb Z}\sum_{x\in J_N}
\mathbb E[\eta_x^Nq_x\psi_{x'}]  K_N(x,x')
\\&=C_{\varphi,n} O\left(\delta_N^{\frac{1}{2}}NP_N(N\varepsilon_N)^{-n}\varepsilon_N^{-1}\right)
+C_{\varphi,m} O\left(\delta_N^{m} N\varepsilon_N^{-1}\right).
\end{align*}
Therefore
\begin{eqnarray*}
&&\int_{\mathbb R}\int_{\mathbb T}
\hat{\varphi}^*(p)\widehat{\mathcal R}_N^{(0)}
\frac{\gamma}{D_N(p,k)}\,dk\,dp
=I_5^N+C_{\varphi} O\left(\frac{e^{-b_0N\varepsilon_N}}{\delta_N\varepsilon_N}\right)
\\&=&C_{\varphi} O\left(\frac{P_N^{-m_1}}{\delta_N\varepsilon_N}\right)
+C_{\varphi,n}O\left(\delta_N^{\frac{1}{2}}NP_N(N\varepsilon_N)^{-n}\varepsilon_N^{-1}\right)
+C_{\varphi,m}O\left(\delta_N^{m} N\varepsilon_N^{-1}\right)
\\&&{}+C_{\varphi} O\left(\frac{e^{-b_0N\varepsilon_N}}{\delta_N\varepsilon_N}\right).
\end{eqnarray*}
By \ref{H2}, we choose suitable $n$ and $m$ such that it tends to zero, which completes the proof of (ii).
\par
\noindent(iii). 
Take the same cutoff function $\chi$ as in (ii). Arguing similarly as  that in (ii), we have 
\begin{eqnarray*}
&&\int_{\mathbb{R}}\int_{\mathbb{T}} \hat{\varphi}^*(p)\widehat{R}_N^{(0)}(p,k)
\int_{\mathbb{T}} \frac{\gamma}{D_N(p,k')} \, dk' \, dk \, dp \\
&=& I_6^N
+ O\left(\frac{e^{-b_0N\varepsilon_N}}{\delta_N\varepsilon_N}\right)
+ O\left(\frac{C_\varphi P_N^{-m_1}}{\delta_N\varepsilon_N}\right),
\end{eqnarray*}
where
\begin{align*}
I_6^N
&= \int_{\mathbb{R}}\int_{\mathbb{T}} \hat{\varphi}^*(p)\widehat{R}_N^{(0)}(p,k)
\int_{\mathbb{T}} \left( \frac{\gamma}{D_N(p,k')} - 1 \right)
\chi\left(\frac{p}{P_N}\right) \, dk' \, dk \, dp .
\end{align*}
Using \eqref{1},
\begin{align*}
\int_{\mathbb{T}} \widehat{R}_N^{(0)}(p,k) \, dk
&= a_N\varepsilon_N \sum_{x'\in\mathbb{Z}}\sum_{x\in J_N}
\mathbb{E}\left[\eta_x^N q_x \psi_{x'}\right]
e^{-\pi i\varepsilon_N p(x+x')}
\int_{\mathbb{T}} e^{2\pi i k(x-x')} \, dk \\
&= a_N\varepsilon_N \sum_{x\in J_N}
\mathbb{E}\left[\eta_x^N q_x \psi_x\right] e^{-2\pi i\varepsilon_N p x}.
\end{align*}
For any $n\ge 2$, integrating by parts $n$ times in $p$ gives

\begin{align*}
|I_6^N| &\le a_N\varepsilon_N \sum_{x\in J_N} \mathbb{E}\left[\left|\eta_x^N q_x \psi_x\right|\right] \\
&\quad \times \int_{\mathbb{T}}\int_{\mathbb{R}} \frac{1}{|\varepsilon_N x|^n} \left| \partial_p^n\left[ \hat{\varphi}^*(p)\left( \frac{\gamma}{D_N(p,k')} - 1 \right) \chi\left(\frac{p}{P_N}\right) \right] \right| \, dp \, dk' \\
&\lesssim a_N\varepsilon_N \sum_{x\in J_N} \left( \mathbb{E}\left[\left|\eta_x^N q_x^2\right|\right] + \mathbb{E}\left[\left|\psi_x\right|^2\right] \right) \frac{1}{|\varepsilon_N x|^n} \delta_N N\varepsilon_N P_N \\
&\lesssim \frac{C_{\varphi,n} NP_N}{(N\varepsilon_N)^{n}} .
\end{align*}
By the choice of $P_N$ and \ref{H2}, the right-hand side tends to $0$.
\end{proof}

It remains to control the term involving the auxiliary field $\bar u_N$ in \eqref{eq:wn}. The following proposition shows that this contribution vanishes in the limit.

\begin{proposition}
For each $T>0$ and each $\varphi \in \mathcal C$, 
\begin{equation}\label{guji5}
\lim_{N\to \infty}\sup_{0\le t\le T}
\left|
a_N\gamma
\int_{\mathbb R}\int_{\mathbb T}
\hat{\varphi}^*(p)\bar u_N(t,p,k)\mathcal L\left(\frac{\gamma}{D_N}\right)(p,k) \,dk\,dp
\right|
= 0 .
\end{equation}
\end{proposition}

\begin{proof}
Set
\begin{equation*}
I_7^N(t):=
a_N\gamma
\int_{\mathbb R}\int_{\mathbb T}
\hat{\varphi}^*(p)\bar u_N(t,p,k)\mathcal L\left(\frac{\gamma}{D_N}\right)(p,k) \,dk\,dp.
\end{equation*}
Combining \eqref{piandao uN} and \eqref{piandao uN-}, we obtain
\begin{align}\label{eq:ubar-identity-1}
a_N|\bar\omega_N|^2\bar u_N
&=
\gamma\left(\widehat U_N^{(0)}-\widehat U_N\right)
+\bar\omega_N\left(\widehat U_{N,-}^{(0)}-\widehat U_{N,-}\right) \notag\\
&\quad
+a_N\gamma^2
\mathcal L\left(
\bar u_N-\frac12\left(\bar w_N+\bar w_{N,-}\right)
\right)
+\gamma\bar r_N^{(3)}
+\bar\omega_N\bar r_N^{(4)} .
\end{align}
Moreover, adding the equations for \eqref{piandao wN} and \eqref{piandao wN-}, gives
\begin{eqnarray} 
&&2a_N\gamma
\mathcal L\left(
\bar u_N-\frac12\left(\bar w_N+\bar w_{N,-}\right)
\right)\nonumber\\
&=&
-i a_N\Delta_N\omega\left(\bar w_N-\bar w_{N,-}\right) \notag\\
&&\quad 
+\widehat W_N^{(0)}
+\widehat W_{N,-}^{(0)}
-\left(\widehat W_N+\widehat W_{N,-}\right)
+\bar r_N^{(1)}+\bar r_N^{(2)}. \label{eq:ubar-identity-3}
\end{eqnarray}
Substituting \eqref{eq:ubar-identity-3} into \eqref{eq:ubar-identity-1}, we get
\begin{eqnarray}\label{eq:ubar-expanded}
&&a_N\bar u_N -
\frac{1}{|\bar\omega_N|^2}
\left(
\gamma\bar r_N^{(3)}
+\bar\omega_N\bar r_N^{(4)}
+\frac{\gamma}{2}\bar r_N^{(1)}
+\frac{\gamma}{2}\bar r_N^{(2)}
\right) \notag\\
&=&
\frac{1}{|\bar\omega_N|^2}
\Big[
\gamma\left(\widehat U_N^{(0)}-\widehat U_N\right)
+\bar\omega_N\left(\widehat U_{N,-}^{(0)}-\widehat U_{N,-}\right)
+\frac{\gamma}{2}
\Big(
\widehat W_N^{(0)}
+\widehat W_{N,-}^{(0)}
\notag\\
&&\quad-\widehat W_N
-\widehat W_{N,-}
\Big)
\Big] 
-\frac{i\gamma}{2|\bar\omega_N|^2}
a_N\Delta_N\omega\left(\bar w_N-\bar w_{N,-}\right).
\end{eqnarray}
By Lemma \ref{lem:omega_N_properties}, we have
\begin{equation*}
2\omega_0\le \bar\omega_N(p,k)\lesssim 1 .
\end{equation*}
By Lemma \ref{lem:yuxiang}, the contribution of the terms containing $\bar r_N^{(i)}$ is negligible. More precisely,
\begin{align}
&\sup_{0\le t\le T}
\left|
\int_{\mathbb R}\int_{\mathbb T}
\hat{\varphi}^*(p)
\bar r_N^{(i)}(t,p,k)
\mathcal L\left(\frac{\gamma}{D_N}\right)(p,k)\,dk\,dp
\right| \lesssim TC_\varphi o(1),
\end{align}
for $i\in\{1,2,3,4\}$. \par
We first consider the region $\{|p|> P_N\}$.
By \eqref{def:花体L} and \eqref{le:DNbounded},  
\begin{equation*}
\left|
\mathcal L\left(\frac{\gamma}{D_N}\right)(p,k)
\right|
\lesssim 1
\end{equation*}
and by Remark \ref{remark:粗略估计}, 
\begin{equation}
\sup_{0\le t\le T}\sup_{p\in\mathbb R}
\int_{\mathbb T}\left|\bar u_N(t,p,k)\right|+\frac{1}{a_N}\left|\bar r_N^{(i)}(t,p,k) \right|\,dk
\lesssim T N\varepsilon_N.
\end{equation}
Then we have 
\begin{eqnarray*}
&&\sup_{0\le t\le T}
\left|
\gamma
\int_{\{|p|>P_N\}}\int_{\mathbb T}
\hat{\varphi}^*(p)\left(\left| a_N\bar u_N(t,p,k)\right|+\left|\bar r_N^{(i)}(t,p,k) \right|\right)
\mathcal L\left(\frac{\gamma}{D_N}\right)(p,k)\,dk\,dp
\right| \notag\\
& \lesssim&
T C_\varphi a_N N\varepsilon_N P_N^{-m_1}
= T C_\varphi O\left(\frac{P_N^{-m_1}}{\delta_N\varepsilon_N}\right)\rightarrow 0\,,\quad N\rightarrow\infty.
\end{eqnarray*}
Therefore, it is enough to estimate
\begin{align*}
I_7^N(t)
&:=
-\gamma
\int_{\{|p|\le P_N\}}\int_{\mathbb T}
\hat{\varphi}^*(p)
\left[
a_N\bar u_N
-
\frac{1}{|\bar\omega_N|^2}
\left(
\gamma\bar r_N^{(3)}
+\bar\omega_N\bar r_N^{(4)}
+\frac{\gamma}{2}\bar r_N^{(1)}
+\frac{\gamma}{2}\bar r_N^{(2)}
\right)
\right] \notag\\
&\quad\ 
\times
\mathcal L\left(\frac{\gamma}{D_N}\right)(p,k)\,dk\,dp .
\end{align*}

For $|p|\le P_N$, by \eqref{eq:DNtalor}, we have
\begin{equation}\label{eq:L-D-expansion}
\mathcal L\left(\frac{\gamma}{D_N}\right)(p,k)
=
\frac{i\varepsilon_Np\omega_N'(k)}{\gamma}
+
O\left(\left(\delta_NN\varepsilon_N|p|\right)^2\right).
\end{equation}
On the other hand, by \eqref{def:UN}, \eqref{eq:R1}--\eqref{eq:R4}, the expression
\begin{equation}
a_N\bar u_N
-
\frac{1}{|\bar\omega_N|^2}
\left(
\gamma\bar r_N^{(3)}
+\bar\omega_N\bar r_N^{(4)}
+\frac{\gamma}{2}\bar r_N^{(1)}
+\frac{\gamma}{2}\bar r_N^{(2)}
\right)
\end{equation}
is an even function of $k$. Since $\omega_N'(k)$ is odd, we have
\begin{align*}
&\int_{\mathbb T}
\left[
a_N\bar u_N
-
\frac{1}{|\bar\omega_N|^2}
\left(
\gamma\bar r_N^{(3)}
+\bar\omega_N\bar r_N^{(4)}
+\frac{\gamma}{2}\bar r_N^{(1)}
+\frac{\gamma}{2}\bar r_N^{(2)}
\right)
\right]
\frac{i\varepsilon_Np\omega_N'(k)}{\gamma}\,dk
=0 .
\end{align*}
Consequently,
\begin{align*}
I_7^N(t)
&=
-\gamma
\int_{\{|p|\le P_N\}}\int_{\mathbb T}
\hat{\varphi}^*(p)
\left[
a_N\bar u_N
-
\frac{1}{|\bar\omega_N|^2}
\left(
\gamma\bar r_N^{(3)}
+\bar\omega_N\bar r_N^{(4)}
+\frac{\gamma}{2}\bar r_N^{(1)}
+\frac{\gamma}{2}\bar r_N^{(2)}
\right)
\right] \notag\\
&\quad
\times
\left[
\mathcal L\left(\frac{\gamma}{D_N}\right)(p,k)
-\frac{i\varepsilon_Np\omega_N'(k)}{\gamma}
\right]\,dk\,dp .
\end{align*}

By Remark \ref{remark:粗略估计}, we have
\begin{align*}
&\left\|\widehat W_N(t,p,\cdot)\right\|_{L^1(\mathbb T)}
+\left\|\widehat W_{N,-}(t,p,\cdot)\right\|_{L^1(\mathbb T)}
+\left\|\widehat U_N(t,p,\cdot)\right\|_{L^1(\mathbb T)}
+\left\|\widehat U_{N,-}(t,p,\cdot)\right\|_{L^1(\mathbb T)}
\notag\\
&
+\left\|\bar w_N(t,p,\cdot)\right\|_{L^1(\mathbb T)}
+\left\|\bar w_{N,-}(t,p,\cdot)\right\|_{L^1(\mathbb T)}
\lesssim T N\varepsilon_N .
\end{align*}
Moreover,
\begin{equation*}
\Delta_N\omega(p,k)=O\left(\delta_NN\varepsilon_N|p|\right).
\end{equation*}
Using \eqref{eq:ubar-expanded}, we therefore get
\begin{align*}
&\sup_{0\le t\le T}
\int_{\mathbb T}
\left|
a_N\bar u_N
-
\frac{1}{|\bar\omega_N|^2}
\left(
\gamma\bar r_N^{(3)}
+\bar\omega_N\bar r_N^{(4)}
+\frac{\gamma}{2}\bar r_N^{(1)}
+\frac{\gamma}{2}\bar r_N^{(2)}
\right)
\right|\,dk
\notag\\
\lesssim&
N\varepsilon_N+TN|p|.
\end{align*}
Combining this bound with \eqref{eq:L-D-expansion}, we obtain
\begin{align*}
\sup_{0\le t\le T}\left|I_7^N(t)\right|
\quad&\lesssim\quad
\int_{\{|p|\le P_N\}}
|\hat{\varphi}(p)|
\left(N\varepsilon_N+TN|p|\right)
\left(\delta_NN\varepsilon_N|p|\right)^2\,dp \notag\\
&\lesssim\quad
\delta_N^2N^2\varepsilon_N^2
\left[
N\varepsilon_N
\int_{\mathbb R}|\hat{\varphi}(p)||p|^2\,dp
+
TN
\int_{\mathbb R}|\hat{\varphi}(p)||p|^3\,dp
\right] \notag\\
&\le\quad
T C_\varphi O\left(\delta_N^2N^3\varepsilon_N^2\right).
\end{align*}
By \ref{H2}, this term tends to $0$.

Combining the estimate of the tail region, the negligible remainder terms, and the estimate of $I_7^N(t)$, we conclude that
\begin{align*}
\sup_{0\le t\le T}\left|I_7^N(t)\right|
&\le
T C_\varphi O\left(\frac{P_N^{-m_1}}{\delta_N\varepsilon_N}\right)
+
T C_\varphi O\left(\delta_N^2N^3\varepsilon_N^2\right)
+T C_\varphi o(1).
\end{align*}
It tends to zero as $N \to \infty$, which proves \eqref{guji5}.
\end{proof}

For any $\varphi \in \mathcal{C}$, by multiplying both sides of \eqref{eq:wn} by a test function $\hat{\varphi}^*$ and integrating over $p$, we can consolidate the results from \eqref{guji1}, \eqref{guji2}, \eqref{guji3}, \eqref{guji4}, \eqref{guji5} and \ref{H2} to rewrite \eqref{eq:wn}  into the following equivalent weak form.
\begin{align} \label{eq:wnlast}
   \langle \widehat{W}_N(t), \hat{\varphi} \rangle &= \langle \mu_0, \varphi \rangle - \frac{\hat{c}}{\gamma} \int_0^t \int_{\mathbb{R}} \int_{\mathbb{T}} p^2 \hat{\varphi}^*(p) \widehat{W}_N(s,p,k) \, dk \, dp \, ds + T R_N^\varphi,
\end{align}
where $R_N^\varphi$ denotes a remainder term depending on $\varphi$ that tends to zero.
Using the property of the Fourier transform 
\begin{align*}
p^2 \hat{\varphi}^*(p) &= p^2 \int_{\mathbb{R}} e^{2\pi i p x} \varphi(x) \, dx 
= -\frac{1}{4\pi^2} \int_{\mathbb{R}} e^{2\pi i p x} \varphi''(x) \, dx 
= -\frac{1}{4\pi^2} \widehat{\varphi''}^*(p).
\end{align*}
Substituting this back yields the fundamental relation.
\begin{equation} \label{eq:star2}
\langle W_N(t), \varphi \rangle - \langle \mu_0, \varphi \rangle - \frac{\hat{c}}{4\pi^2\gamma} \int_0^t \langle W_N(s), \varphi'' \rangle \, ds = TR_N^\varphi.
\end{equation}

\section{Proof of Theorem \ref{mainthm}}
\begin{proof}
For any $t \ge 0$, we define the empirical measure $\mu_N(t) \in \mathcal{M}_+(\mathbb{R})$ as
$$
\mu_N(t) = \frac{\varepsilon_N}{2} \sum_{x \in \mathbb{Z}} \mathbb{E}\left[|\psi_x(a_N t)|^2\right] \delta_{\varepsilon_N x}.
$$
By the Plancherel theorem, we have $\langle \mu_N(t), \varphi \rangle = \langle \widehat{W}_N(t), \hat{\varphi} \rangle = \langle W_N(t), \varphi \rangle$. 

We endow $\mathcal M_+(\mathbb R)$ with the vague topology. By the Weierstrass approximation theorem and smooth truncation via cutoff functions, there exists a countable family $\{\varphi_m\} \subset C_c^\infty(\mathbb{R})$ satisfying the following property: for any $\varphi \in C_c(\mathbb{R})$, there exist a compact set $K'$ and a subsequence $\{\varphi_{k_j}\}$ such that $\operatorname{supp}(\varphi_{k_j}) \subset K'$, $\operatorname{supp}(\varphi) \subset K'$, and
$ \lim_{j \to \infty} \|\varphi_{k_j} - \varphi\|_\infty = 0$ (see \cite[Section~A.10]{seppalainen2008translation}).
For any $\mu,\mu'\in\mathcal M_+(\mathbb R),$ we define the metric $d(\mu, \mu')$ as
\begin{align}
    d(\mu, \mu') = \sum_{m=1}^\infty 2^{-m} \min\left\{1, \left|\mu(\varphi_m) - \mu'(\varphi_m)\right|\right\}.
\end{align}
It can be proved that $(\mathcal{M}_+(\mathbb{R}), d)$ is a Polish space and that the topology induced by the metric $d$ coincides with the vague topology. In other words, $d(\mu_n,\mu)\to0$ if and only if $\mu_n(\varphi)\to\mu(\varphi)$ for every $\varphi\in C_c(\mathbb{R})$ (see \cite[Section~A.10]{seppalainen2008translation}).
\par
Next, we apply the Arzelà--Ascoli Theorem~(\cite[Theorem~6.3.1, e.g.]{dixmierGeneralTopology1984} )
to prove that $\{\mu_N\}$ is relatively compact in $C([0,T], \mathcal{M}_+(\mathbb{R}))$ under the vague topology.

First, we prove pointwise relative compactness. For any $T > 0$ and any compact set $K \subset \mathbb{R}$, suppose $K \subset [-R_0, R_0]$. Let $\phi(x) = e^{-b_0\sqrt{1+x^2}} \in \mathcal{C}$. By direct differentiation, we obtain the bound $|\phi''| \le C_{\phi}\phi(x)$. 

Using \eqref{eq:star2}, we can estimate
\begin{align*}
\langle \mu_N(t), \phi \rangle &= \langle \mu_0, \phi \rangle + \frac{\hat{c}}{4\pi^2\gamma} \int_0^t \langle \mu_N(s), \phi'' \rangle \, ds + T R_N^\phi \\
&\le \frac{\hat{c}C_{\phi}}{4\pi^2\gamma} \int_0^t \langle \mu_N(s), \phi \rangle \, ds + \langle \mu_0, \phi \rangle + T R_N^\phi.
\end{align*}
By Gronwall's inequality, we deduce
$$
\langle \mu_N(t), \phi \rangle \le \left( \langle \mu_0, \phi \rangle + T R_N^\phi \right) e^{\frac{\hat{c}C_{\phi}}{4\pi^2\gamma} T}.
$$
Notice that (H3) ensures that $\langle \mu_0, \phi \rangle < \infty$ and  since
$$
\mu_N(t)(K) \le e^{b_0\sqrt{1+R_0^2}} \langle \mu_N(t), \phi \rangle,
$$
we conclude that $\sup_N \sup_{0 \le t \le T} \mu_N(t)(K) < \infty$. By the compactness theorem for positive Radon measures, the uniform local mass bound ensures that for each fixed $t \in [0,T]$, $\{\mu_N(t)\}$ is relatively compact in $\mathcal{M}_+(\mathbb{R})$.

Next, we prove equicontinuity. For any $T > 0$, $0 \le s, t \le T$, \eqref{eq:star2} yields
\begin{align*}
\left| \langle \mu_N(t) - \mu_N(s), \varphi_m \rangle \right| &= \left| \frac{\hat{c}}{4\pi^2\gamma} \int_s^t \langle \mu_N(u), \varphi_m'' \rangle \, du + T R_N^{\varphi_m} \right| \\
&\le |t-s| \frac{\hat{c}}{4\pi^2\gamma} \sup_N \sup_{0 \le t \le T} |\langle \mu_N(t), \varphi_m'' \rangle| + T R_N^{\varphi_m}.
\end{align*}
Since $\varphi_m'' \in C_c^\infty(\mathbb{R})$, the uniform boundedness implies
\begin{align*}
    \sup_N \sup_{0 \le t \le T} |\langle \mu_N(t), \varphi_m'' \rangle|
    =\widehat C_m<\infty.
\end{align*}
For any given $\varepsilon' > 0$, we can choose an integer $m_2 > 1$ such that $\sum_{m > m_2} 2^{-m} < \varepsilon'$. Then, we have
\begin{align*}
    d(\mu_{N(t)}, \mu_{N(s)}) 
    &= \sum_{m=1}^\infty 2^{-m} \min\left\{1, \left|\langle \mu_{N(t)} - \mu_{N(s)}, \varphi_m \rangle\right|\right\} \\
    &\le \sum_{m=1}^{m_2} 2^{-m} \left( |t-s| \frac{\hat{c}}{4\pi^2\gamma} \widehat{C}_m + T R_N^{\varphi_m} \right) + \sum_{m > m_2} 2^{-m} \\
    &\le \frac{\hat{c}}{4\pi^2\gamma} (\widehat{C}_1 + \dots + \widehat{C}_{m_2}) |t-s| + T \left(R_N^{\varphi_1} + \dots + R_N^{\varphi_{m_2}}\right) + \varepsilon'.
\end{align*}
Therefore, 
\begin{align}
    \lim_{\delta \to 0} \limsup_{N \to \infty} \sup_{\substack{0 \le s, t \le T \\ |t-s| \le \delta}} d(\mu_{N(t)}, \mu_{N(s)}) \le \varepsilon'.
\end{align}
Since $\varepsilon' > 0$ is arbitrary, we conclude that the sequence $\{\mu_N\}$ is equicontinuous.

By the Arzelà--Ascoli Theorem, $\{\mu_N\}$ is relatively compact in $C([0,T], \mathcal{M}_+(\mathbb{R}))$. Therefore, there exists a subsequence $\{\mu_{N_j}\}$ and  $\mu \in C([0,T], \mathcal{M}_+(\mathbb{R}))$ such that for any $\varphi \in C_c(\mathbb{R})$,
$$
\lim_{j \to \infty} \sup_{0 \le t \le T} \left| \langle \mu_{N_j}(t) - \mu(t), \varphi \rangle \right| = 0.
$$
For any $\varphi\in C_c^2(\mathbb R)$, passing to the limit $N_j \to \infty$ along the subsequence in \eqref{eq:star2}, we obtain
\begin{align}
    \langle \mu(t), \varphi \rangle = \langle \mu_0, \varphi \rangle + \frac{\hat{c}}{4\pi^2\gamma} \int_0^t \langle \mu(s), \varphi'' \rangle \, ds,
\end{align}
which means $\mu(t)$ is a weak solution of \eqref{measure heat equation}.

Due to the uniqueness of the weak solution to the heat equation, the entire sequence converges. That is, for any $\varphi \in C_c(\mathbb{R})$,
$$
\lim_{N \to \infty} \sup_{0 \le t \le T} \left| \langle \mu_N(t) - \mu(t), \varphi \rangle \right| = 0.
$$

Finally, applying Proposition \ref{thm:wavebijin}, we conclude that for any $\varphi \in C_c(\mathbb{R})$,
$$
\lim_{N \to \infty} \sup_{0 \le t \le T} \left| \varepsilon_N \sum_{x \in \mathbb{Z}} \varphi(\varepsilon_N x) \mathbb{E}[e_x(a_N t)] - \int_{\mathbb{R}} \varphi(y) \mu(t, dy) \right| = 0,
$$
where $\mu(t, dy)$ is the weak solution of \eqref{measure heat equation}.
\end{proof}

\appendix

\section{Schwartz Function and Fourier Transform}
In this appendix, we introduce some definitions and properties of Schwartz functions and Fourier transforms. The Schwartz space $\mathcal{S}(\mathbb{R})$ is defined as the set of all smooth functions whose derivatives of all orders decay faster than the inverse of any polynomial, that is 
\begin{align}
    \mathcal{S}(\mathbb{R}) := \bigg\{ \varphi \in C^\infty(\mathbb{R}) : \sup_{x \in \mathbb{R}} (1+|x|)^m |\partial^n \varphi(x)| < \infty, \quad \text{for   every} \ m, n \in \mathbb{N} \bigg\}.
\end{align}
For any $\varphi \in \mathcal{S}(\mathbb{R})$, we denote its associated seminorms by
\begin{equation}
    \|\varphi\|_{m, n} := \sup_{x \in \mathbb{R}} (1+|x|)^m |\partial^n \varphi(x)|.
\end{equation}
So for any $\varphi\in \mathcal{S}(\mathbb{R}),\ m \in \mathbb{N},\ y\in\mathbb{R},$ we have $|\varphi(y)| \lesssim C_{\varphi,m} (1+|y|)^{-m}.$
\par
Let $\mathbb{T} = \mathbb{R}/\mathbb{Z}$ be the one-dimensional torus, which can also be identified with $\left[-\frac{1}{2}, \frac{1}{2}\right]$. Let $\mathcal{F}$ and $\mathcal{F}^{-1}$ denote the Fourier transform and its inverse, respectively.
For functions $f\in \ell^2(\mathbb{Z})$, 
 and $\varphi \in \mathcal{S}(\mathbb{R})$, we define the Fourier transform and its inverse
 \begin{equation}
     \begin{aligned}
    \hat{f}(k) &:= (\mathcal{F}f)(k)=\sum_{x \in \mathbb{Z}} f_x e^{-2\pi i k x}, \quad k \in \mathbb{T}, \\
    (\mathcal{F}^{-1}\hat{f})_x&:= \int_{\mathbb{T}} \hat{f}(k) e^{2\pi i k x} \, dk, \quad x \in \mathbb{Z}, \\
    \hat{\varphi}(p) &:=(\mathcal{F}\varphi)(p)= \int_{\mathbb{R}} \varphi(x) e^{-2\pi i p x} \, dx, \quad p \in \mathbb{R}, \\
     (\mathcal{F}^{-1}\hat\varphi)(y)&:=\int_{\mathbb{R}} \hat\varphi(p) e^{2\pi i p y} \, dp, \quad y \in \mathbb{R}.
\end{aligned}
 \end{equation}
The Fourier transform is an automorphism on $\mathcal{S}(\mathbb{R})$, and $\varphi=\mathcal{F}^{-1}\hat\varphi,\ f=\mathcal{F}^{-1}\hat{f}$.

For $f, g \in \ell^2(\mathbb{Z})$, the discrete convolution is defined by
\begin{align}
    (f * g)_x := \sum_{y \in \mathbb{Z}} f_{x-y} g_y.
\end{align}
 Direct computation yields $\widehat{f \ast g} = \hat{f} \hat{g}$ and $\widehat{fg} = \hat{f} \ast \hat{g}$ and for $f, g \in \ell^2(\mathbb{Z})$, the Plancherel equality  holds
\begin{align}
    \sum_{x \in \mathbb{Z}} f_x g_x^* = \int_{\mathbb{T}} \hat{f}_k \hat{g}_k^* \, dk.
\end{align}
In particular, when $f = g$, we have $\|f\|_{\ell^2(\mathbb{Z})} = \|\hat{f}\|_{L^2(\mathbb{T})}$.

For Schwartz function $f \in \mathcal{S}(\mathbb{R})$ and  $p \in \mathbb{R}$, the Poisson summation formula holds~\cite[Theorem~3.1.17]{grafakosClassicalFourierAnalysis2008}
\begin{align}
    \sum_{x \in \mathbb{Z}} \hat{f}(x) e^{2\pi i p x} = \sum_{x \in \mathbb{Z}} f(p + x).
\end{align}
Note that both sides of the identity above are periodic functions of $p$ with period~$1$. By replacing $f$ with its inverse Fourier transform $\check{f}$, we obtain for any $k\in\mathbb{T},$
\begin{align} \label{eq:Poisson summation formula2}
    \sum_{x \in \mathbb{Z}} f(x) e^{2\pi i k x} = \sum_{x \in \mathbb{Z}} \hat{f}(k + x).
\end{align}

\begin{lemma} \label{lem:Fourier Transform}
For all real-valued function 
$\varphi \in \mathcal{C}$ and  $f, g \in \ell^2(\mathbb{Z})$
\begin{align}
    \sum_{x \in \mathbb{Z}} \varphi(\varepsilon_N x) f_x g_x^* = \int_{\mathbb{R} \times \mathbb{T}} \hat{\varphi}^*(p) \hat{f}\Big(k + \frac{\varepsilon_N p}{2}\Big) \hat{g}^*\Big(k - \frac{\varepsilon_N p}{2}\Big) \, dp \, dk.
\end{align}
\end{lemma}

\begin{proof}
By the definition of the Fourier transform and Fubini's Theorem, we expand the right-hand side as  
\begin{align*}
     &\int_{\mathbb{R} \times \mathbb{T}} \hat{\varphi}^*(p) \hat{f}\Big(k + \frac{\varepsilon_N p}{2}\Big) \hat{g}^*\Big(k - \frac{\varepsilon_N p}{2}\Big) \, dp \, dk \\
    &= \int_{\mathbb{R} \times \mathbb{T}} \hat{\varphi}^*(p) \bigg( \sum_{x \in \mathbb{Z}} f_x e^{-2\pi i (k + \frac{\varepsilon_N p}{2}) x} \bigg) \bigg( \sum_{x' \in \mathbb{Z}} g_{x'}^* e^{2\pi i (k - \frac{\varepsilon_N p}{2}) x'} \bigg) \, dp \, dk \\
    &= \sum_{x, x' \in \mathbb{Z}} f_x g_{x'}^* \bigg( \int_{\mathbb{R}} \hat{\varphi}^*(p) e^{-\pi i \varepsilon_N p (x + x')} \, dp \bigg) \bigg( \int_{\mathbb{T}} e^{2\pi i k (x' - x)} \, dk \bigg).\\
    &= \sum_{x, x' \in \mathbb{Z}} f_x g_{x'}^*  \varphi^*\bigg(\frac{\varepsilon_N(x+x')}{2}\bigg) \delta_{x, x'} \\
    &=\sum_{x \in \mathbb{Z}} \varphi^*(\varepsilon_N x)f_x g_x^*
    = \sum_{x \in \mathbb{Z}} \varphi(\varepsilon_N x)f_x g_x^*.
\end{align*}
\end{proof}

\section{Auxiliary Results}

\begin{lemma} \label{lem:eta}
The quantity $\eta_x^N$ defined by \eqref{def:etaN} satisfies  
\begin{enumerate}[label=\textup{(\roman*)}, align=left]
    \item  $\ \sup_{x,N} \eta_x^N<\infty$.
    \item If $|x|\leq N-M\delta_N(2N+1)$, then $\eta_x^N=0$\,.
\end{enumerate}
    \begin{proof}
Part (i) follows immediately from the definition of $\eta_x^N$ and Lemma \ref{lem:alpha_properties}~(iii) with $m=0$. For part (ii), since $\alpha^N$ is even, we have
\begin{align}
    \eta_x^N = \sum_{y=-N}^{N} \alpha_{x-y}^N - \omega_0^2 = \sum_{y=-N-x}^{N-x} \alpha_y^N - \omega_0^2. 
\end{align}
For $|x| \leq N - M \delta_N(2N+1)$, the summation limits satisfy $N-x \geq M \delta_N(2N+1)$ and $-N-x \leq -M \delta_N(2N+1)$. Since the support of $\alpha^N$ satisfies $\operatorname{supp}(\alpha^N) \subset [-M \delta_N(2N+1), M \delta_N(2N+1)]$, the summation covers the entire support of $\alpha^N$. Recalling that $\sum_{y \in \mathbb{Z}} \alpha_y^N = \omega_0^2$, we conclude that $\eta_x^N = 0$ in this region.
    \end{proof}
\end{lemma}
Now let 
\begin{equation} \label{def:JN}
    J_N := \{ x \in \{-N, \dots, N\} : |x| > N - M \delta_N (2N+1) \}.
\end{equation}
Then  $\eta_x^N$ is non-zero just for $x \in J_N$.

\begin{lemma} \label{lem:alpha_properties}
The sequence $\alpha^N$ defined by \eqref{def:alphaN} satisfy  
\begin{enumerate}[label=\textup{(\roman*)}, leftmargin=*, itemsep=0.3em]
  \item For any $k\in \mathbb{T},\ $
  $\displaystyle \widehat{\alpha^N}(k)\in\mathbb{R},\ 
  \widehat{\alpha^N}(k)=\widehat{\alpha^N}(-k)\geq \omega_0^2$.
 \item  As $N \to \infty$,  
\[
 \sum_{x \neq 0} |\alpha_x^N| = -\sum_{x \neq 0} \alpha_x^N 
 \to \lVert\alpha\rVert_{L^1(\mathbb{R})}.
\]
  \item For each $m \in \mathbb{N}$, there exist constants $C_m > 0$ independent of $N$ such that
  \[
 \sum_{x \in \mathbb{Z}} |x|^m |\alpha_x^N|  =  C_m(\delta_N N)^m+O((\delta_N N)^{m-1}).
\]
\item  
$\displaystyle \int_{\mathbb{T}} \left| (\widehat{\alpha^N})'(k)
 \right|^2\, dk 
\asymp \delta_N N$.
\end{enumerate}
\end{lemma}

\begin{proof}
Let $L_N=\delta_N(2N+1)$.\\
\noindent (i). By the definition of the Fourier transform, we have
\[ 
\widehat{\alpha^N}(k) = \sum_{x \in \mathbb{Z}} \alpha_x^N e^{-2\pi i k x} = \sum_{x \in \mathbb{Z}} \alpha_x^N (\cos(2\pi k x) - i \sin(2\pi k x)). 
\]
Since the sequence $\alpha_x^N$ is even, the imaginary part vanishes which implies that $\widehat{\alpha^N}(k) \in \mathbb{R}$ and $\widehat{\alpha^N}(k) = \widehat{\alpha^N}(-k)$.

Specifically, separating the term at $x=0$, we obtain
\[ 
\widehat{\alpha^N}(k) = \omega_0^2 + \frac{1}{L_N} \sum_{x \neq 0} \alpha\left(\frac{x}{L_N}\right) (\cos(2\pi k x) - 1). 
\]
Since $\alpha(y) \leq 0$ and $\cos(2\pi k x) - 1 \leq 0$, each term in the sum is nonnegative. Therefore, we conclude that $\widehat{\alpha^N}(k) \geq \omega_0^2$.

\noindent (ii). For $x \neq 0$, by definition $\alpha_x^N \leq 0$, which gives $|\alpha_x^N| = -\alpha_x^N$. Thus,
\[ 
\sum_{x \neq 0} |\alpha_x^N| = -\sum_{x \neq 0} \alpha_x^N = -\frac{1}{L_N} \sum_{x \neq 0} \alpha\left(\frac{x}{L_N}\right). 
\]
Recognizing this as a Riemann sum, as $N \to \infty$, it converges to the integral
\[ 
-\int_{\mathbb{R}} \alpha(x) dx = \int_{\mathbb{R}} |\alpha(x)| dx = \lVert\alpha\rVert_{L^1(\mathbb{R})}. 
\]

\noindent (iii).  For $m \in \mathbb{N}$, by a Riemann sum approximation, we have
\begin{align*} 
\sum_{x \in \mathbb{Z}} |x|^m |\alpha_x^N| &= (L_N)^m \left[ \frac{1}{L_N} \sum_{x \neq 0} \left| \frac{x}{L_N} \right|^m \left| \alpha\left(\frac{x}{L_N}\right) \right| \right] \\ 
&= (L_N)^m \left( \int_{\mathbb{R}} |x|^m |\alpha(x)| dx + O\left(\frac{1}{\delta_N N}\right) \right) \\ 
&= C_m (\delta_N N)^m + O((\delta_N N)^{m-1}), 
\end{align*}
where $C_m := \int_{\mathbb{R}} |x|^m |\alpha(x)| dx > 0$.

\noindent (iv). Direct  computation yields 
\[ 
\int_{\mathbb{T}} \left| (\widehat{\alpha^N})'(k) \right|^2 dk 
= 4\pi^2 \sum_{x \in \mathbb{Z}} |x|^2 |\alpha_x^N|^2. 
\]
Again, recognizing the Riemann sum, we have
\begin{align*}
    \sum_{x \in \mathbb{Z}} |x|^2 |\alpha_x^N|^2 
    &= \frac{1}{(\delta_N L_N)^2} \sum_{x \neq 0} |x|^2 \left| \alpha\left(\frac{x}{\delta_N L_N}\right) \right|^2 
    \\&= \delta_N L_N \left[ \frac{1}{\delta_N L_N} \sum_{x \neq 0} \left| \frac{x}{\delta_N L_N} \right|^2 \left| \alpha\left(\frac{x}{\delta_N L_N}\right) \right|^2 \right]
 \\&=\delta_N L_N\int_{\mathbb{R}} |x|^2 |\alpha(x)|^2 dx+O\left(1\right)
 \asymp \delta_N N.
\end{align*}
\end{proof}

\begin{remark} \label{remark:alpha}
  Regarding the previous lemma we make the following observations.
\begin{enumerate}[label=\textup{(\arabic*)}, leftmargin=*, itemsep=0.3em]
    \item By (ii), $\alpha_0^N \to \omega_0^2 + \lVert\alpha\rVert_{L^1(\mathbb{R})} $.
    \item By (i) and (ii), for any $k\in \mathbb{T},\ \widehat{\alpha^N}(k)\asymp 1$.
    \item By (iii), 
    \[
   \sup_{k\in\mathbb{T}} \left| \frac{d^m}{dk^m} \widehat{\alpha^N}(k)\right| \lesssim C_m (\delta_N N)^m .
    \]
\end{enumerate}
\end{remark}

\begin{lemma} \label{lem:omega_N_properties}
For all $k\in \mathbb{T},\ $ $\omega_N(k)$ defined by \eqref{def:omega} satisfies 
\begin{enumerate}[label=\textup{(\roman*)}, leftmargin=*, itemsep=0.5em]
    \item  $\omega_N(k)\in\mathbb{R}$,\ 
    $\omega_N(k) = \omega_N(-k)$ and $\check{\omega}_N(-x) = \check{\omega}_N(x)$.
    
    \item 
    $\displaystyle   \int_{\mathbb{T}} |\omega_N'(k)|^2 \,dk \asymp  \delta_N N $.
    \item  For  $m \in \mathbb{N}$,
    \begin{align*}
         \sup_{k\in\mathbb{T}}\left| \frac{d^m}{dk^m} \omega_N(k) \right| =C_mO( (\delta_N N)^m). 
    \end{align*}
\end{enumerate}
\end{lemma}

\begin{proof}
\noindent \textbf{(i)} By Lemma \ref{lem:alpha_properties}, $\widehat{\alpha^N}(k) = \widehat{\alpha^N}(-k)$. It immediately follows that $\omega_N(k) = \sqrt{\widehat{\alpha^N}(k)} = \sqrt{\widehat{\alpha^N}(-k)} = \omega_N(-k)$. Consequently, the inverse Fourier transform is symmetric, yielding $\check{\omega}_N(-x) = \check{\omega}_N(x)$.

\vspace{0.5em}
\noindent \textbf{(ii)}
$
|\omega_N'(k)|^2 = \frac{\left| (\widehat{\alpha^N})'(k) \right|^2}{4 \widehat{\alpha^N}(k)}. $
From Lemma \ref{lem:alpha_properties} (iv) and  Remark \ref{remark:alpha} (2).

\vspace{0.5em}
\noindent \textbf{(iii)} We proceed by mathematical induction on $m$. For $m=0$, the bound easily follows from the fact that $\omega_N(k)$ is bounded below and $\widehat{\alpha^N}(k) = O(1)$. 

Assume that the bound holds for all derivatives up to order $m-1$. Since $\widehat{\alpha^N}(k) = \omega_N^2(k)$,   the general Leibniz rule  yields 
\[ 
\frac{d^m}{dk^m} \widehat{\alpha^N}(k) = \sum_{j=0}^{m} \binom{m}{j} \left( \frac{d^j}{dk^j} \omega_N(k) \right) \left( \frac{d^{m-j}}{dk^{m-j}} \omega_N(k) \right). 
\]
Isolating the terms involving the $m$-th derivative (i.e., $j=0$ and $j=m$), we get
\[ 
2 \omega_N(k) \frac{d^m}{dk^m} \omega_N(k) = \frac{d^m}{dk^m} \widehat{\alpha^N}(k) - \sum_{j=1}^{m-1} \binom{m}{j} \left( \frac{d^j}{dk^j} \omega_N(k) \right) \left( \frac{d^{m-j}}{dk^{m-j}} \omega_N(k) \right). 
\]
Taking the absolute value and applying the triangle inequality, along with the bound for $\frac{d^m}{dk^m} \widehat{\alpha^N}(k)$ and the induction hypothesis, we deduce

$$
\left| 2 \omega_N(k) \frac{d^m}{dk^m} \omega_N(k) \right| \leq C_m (\delta_N N)^m + \sum_{j=1}^{m-1} C_{m,j} (\delta_N N)^j (\delta_N N)^{m-j} \leq \widetilde{C}_m (\delta_N N)^m. 
$$
Since $\omega_N(k) \asymp 1$, this completes the proof.  
\end{proof}
Actually, we have a more precise estimate for (ii).
\begin{lemma} \label{lem:aN limit}
    For any $m \in \mathbb{N}_+$,  
    \begin{align}
        \Bigg| \frac{1}{\delta_N(2N+1)} \int_{\mathbb{T}} |\omega_N'(k)|^2 \, dk - \frac{1}{4} \int_{\mathbb{R}} \frac{|\hat{\alpha}'(p)|^2}{\omega_0^2 + \|\alpha\|_{L^1} + \hat{\alpha}(p)} \, dp \Bigg| = O((\delta_N N)^{-m}).
    \end{align}
\end{lemma}

\begin{proof}
Let $L_N=\delta_N(2N+1),\ f_N(x) = \frac{1}{L_N} \alpha\big(\frac{x}{L_N}\big) \in C_c^\infty(\mathbb{R})$. A direct computation of its Fourier transform yields
$$
\widehat{f_N}(p) = \int_{\mathbb{R}} f_N(x) e^{-2\pi i p x} \, dx = \int_{\mathbb{R}} \alpha\Big(\frac{x}{L_N}\Big) e^{-2\pi i (L_N p) \frac{x}{L_N}} \, d\Big(\frac{x}{L_N}\Big) = \hat{\alpha}(L_N p).
$$

Recall the expression for $\widehat{\alpha^N}(k)$. Since $f_N$ is an even function, we can rewrite the sum over $\mathbb{Z}$ and apply the Poisson summation formula \eqref{eq:Poisson summation formula2},
\begin{align*}
    \widehat{\alpha^N}(k) &= \omega_0^2 - \sum_{x \neq 0} f_N(x) + \sum_{x \neq 0} f_N(x) e^{-2\pi i k x} \\
    &= \omega_0^2 - \sum_{x \in \mathbb{Z}} f_N(x) + \sum_{x \in \mathbb{Z}} f_N(x) e^{2\pi i k x} \\
    &= \omega_0^2 - \sum_{x \in \mathbb{Z}} \widehat{f_N}(x) + \sum_{x \in \mathbb{Z}} \widehat{f_N}(k+x) \\
    &= \omega_0^2 + \sum_{x \in \mathbb{Z}} \Big( \hat{\alpha}\big(L_N(k+x)\big) - \hat{\alpha}(L_N x) \Big) \\
    &= \omega_0^2 + \|\alpha\|_{L^1} + \hat{\alpha}(L_N k)+
    \sum_{x\in\mathbb Z \setminus\{0\}} \Big( \hat{\alpha}\big(L_N(k+x)\big) - \hat{\alpha}(L_N x) \Big) .
\end{align*}
Differentiating with respect to $k$, we obtain 
$$(\widehat{\alpha^N})'(k) = L_N \sum_{x \in \mathbb{Z}} \hat{\alpha}'\big(L_N(k+x)\big).$$
\par Since $\hat{\alpha}, \hat{\alpha}' \in \mathcal{S}(\mathbb{R})$, for any $m \in \mathbb{N}_+$, we have $|\hat{\alpha}(y)| + |\hat{\alpha}'(y)| \lesssim C_m |y|^{-m}$. Note that for $k \in \mathbb{T} \equiv [-\frac{1}{2}, \frac{1}{2}]$ and $x\in\mathbb Z \setminus\{0\}$, 
\begin{align*}
    |k+x|\ge |x|-|k|\ge |x|-\frac12\ge\frac{|x|}{2}.
\end{align*}
Therefore, for any $m \geq 2$, 
$$\sum_{x\neq 0}|\widehat{\alpha}(L_N(k+x))|
\leq C_m \sum_{x\neq0} \left( L_N\frac{|x|}{2}\right)^{-m}
\leq C_m^{'}L_N^{-m}\lesssim L_N^{-m+1}.$$
The same argument applies to $\sum_{x\neq0}|\widehat{\alpha}(L_Nx)|$, $\sum_{x\neq 0}|\hat{\alpha}'\big(L_N(k+x)\big)|$.
Choosing the decay orders in the preceding estimates sufficiently large, for any $m\in \mathbb{N_+}$, we arrive at the asymptotic expansions
\begin{align*}
    \widehat{\alpha^N}(k) &= \omega_0^2 + \|\alpha\|_{L^1} + \hat{\alpha}(L_N k) + O(L_N^{-m}), \\
    (\widehat{\alpha^N})'(k) &= L_N \hat{\alpha}'(L_N k) + O(L_N^{-2m}).
\end{align*}
Now we evaluate the integral of $|\omega_N'(k)|^2$. Substituting the expansions above, we have
\begin{align*}
   \frac{1}{L_N}\int_{\mathbb{T}} |\omega_N'(k)|^2 \, dk 
    &= \frac{1}{4} \int_{\mathbb{T}} \frac{|(\widehat{\alpha^N})'(k)|^2}{L_N\widehat{\alpha^N}(k)} \, dk \\
    &= \frac{1}{4} \int_{-\frac{1}{2}}^{\frac{1}{2}} \frac{L_N |\hat{\alpha}'(L_N k)|^2 + O(L_N^{-2m})}{\omega_0^2 + \|\alpha\|_{L^1} + \hat{\alpha}(L_N k) + O(L_N^{-m+1})} \, dk.
\end{align*}
Since $\alpha$ is non-positive even,
\begin{align*}
\|\alpha\|_{L^1}+\hat\alpha(L_N k)
=\int_{\mathbb R}
(-\alpha(x))
\left(1-\cos(2\pi L_N kx)\right)\,dx
\geq0.
\end{align*}
Expanding the denominator, the error term can be decoupled additively
$$
\frac{1}{L_N}\int_{\mathbb{T}} |\omega_N'(k)|^2 \, dk = \frac{1}{4} \int_{-\frac{1}{2}}^{\frac{1}{2}} \frac{L_N |\hat{\alpha}'(L_N k)|^2}{\omega_0^2 + \|\alpha\|_{L^1} + \hat{\alpha}(L_N k)} \, dk + O(L_N^{-m-1}).
$$
Performing the change of variables $p = L_N k$, we get
$$
\frac{1}{L_N}\int_{\mathbb{T}} |\omega_N'(k)|^2 \, dk = \frac{1}{4} \int_{-L_N/2}^{L_N/2} \frac{ |\hat{\alpha}'(p)|^2}{\omega_0^2 + \|\alpha\|_{L^1} + \hat{\alpha}(p)} \, dp + O(L_N^{-m-1}).
$$
Because $\hat{\alpha}' \in \mathcal{S}(\mathbb{R})$ and the denominator is uniformly bounded below by $\omega_0^2 > 0$, the tail of the integral for $|p| > L_N/2$ is bounded by $O(L_N^{-m})$. Extending the domain of integration to $\mathbb{R}$, we conclude that for any $m \in \mathbb{N}_+$,
$$
\Bigg| \frac{1}{L_N} \int_{\mathbb{T}} |\omega_N'(k)|^2 \, dk - \frac{1}{4} \int_{\mathbb{R}} \frac{|\hat{\alpha}'(p)|^2}{\omega_0^2 + \|\alpha\|_{L^1} + \hat{\alpha}(p)} \, dp \Bigg| = O(L_N^{-m}).
$$
\end{proof}

\bibliographystyle{plain}
\bibliography{references}

\end{document}